\documentclass[12pt,draftclsnofoot,onecolumn]{ctextemp_IEEEtran}
\usepackage{amssymb}
\usepackage{amsmath}
\usepackage{amsfonts}
\usepackage{amsthm}
\usepackage{algorithm}
\usepackage{graphicx}
\usepackage{subfigure}
\usepackage{cite}
\usepackage{enumerate}
\usepackage{url}
\usepackage{bm}
\usepackage[noend]{algpseudocode}
\usepackage{algorithmicx,algorithm}
\usepackage{setspace}
\usepackage{cite}
\usepackage{epstopdf}
\usepackage{color}

\def\sgn{\mathop{\rm sgn}}
\newtheorem{myDef}{Definition}
\newtheorem{thm}{Theorem}
\newtheorem{prop}{Proposition}

\makeatother

\begin{document}

\title{\LARGE{Cellular V2X in Unlicensed Spectrum: Harmonious Coexistence with VANET in 5G systems}}
\author{
\IEEEauthorblockN{
{Pengfei Wang},
{Boya Di}, 
{Hongliang Zhang}, 
{Kaigui Bian}, 
{and Lingyang Song}} \\

\vspace{-0.5cm}

\thanks{The authors are with School of Electronics Engineering and Computer Science, Peking University, Beijing, China (email: \{wangpengfei13,diboya,hongliang.zhang,bkg,lingyang.song\}@pku.edu.cn).}
}
\maketitle

\begin{abstract}
With the increasing demand for vehicular data transmission, limited dedicated cellular spectrum becomes a bottleneck to satisfy the requirements of all cellular vehicle-to-everything (V2X) users.
To address this issue, unlicensed spectrum is considered to serve as the complement to support cellular V2X users. In this paper, we study the coexistence problem of cellular V2X users and vehicular ad-hoc network~(VANET) users over the unlicensed spectrum.
To facilitate the coexistence, we design an energy sensing based spectrum sharing scheme, where cellular V2X users are able to access the unlicensed channels fairly while reducing the data transmission collisions between cellular V2X and VANET users.
In order to maximize the number of active cellular V2X users, we formulate the scheduling and resource allocation problem as a two-sided many-to-many matching with peer effects.
We then propose a dynamic vehicle-resource matching algorithm (DV-RMA) and present the analytical results on the convergence time and computational complexity.
Simulation results show that the proposed algorithm outperforms existing approaches in terms of the performance of cellular V2X system when the unlicensed spectrum is utilized.
\end{abstract}

\begin{IEEEkeywords}
Cellular V2X, VANET, Semi-persistent scheduling, Resource allocation, Matching theory
\end{IEEEkeywords}

\newpage

\section{Introduction}%
\par Intelligent transport systems (ITS) have been developed for decades to support a wide variety of safety-critical and traffic-efficient applications. Recently, the solution concept of vehicle-to-everything (V2X) communication has drawn great attention in both industrial and academic fields, including vehicle-to-vehicle (V2V) communication, vehicle-to-infrastructure (V2I) communication, and so on \cite{KAEH-2011}.
 Based on the long term evolution~(LTE) technology, the cellular V2X communication supports massive data transmission in large coverage with controllable latency~\cite{KAEH-2011,SS-2010,FA-2015,SNHH-2015}.
However, due to the high user density of the vehicular network, especially in the dense urban scenario, the dedicated spectrum for cellular V2X communication may not fully satisfy the demands of massive data transmission.
 In addition, requirements of reliable data transmission with low latency is hard to be satisfied, especially in safety-critical applications.

\par To improve the performance of vehicular communication and satisfy the aforementioned demands, several technologies have been discussed in V2X systems.
Multicast-broadcast single-frequency network~(MBSFN) is investigated to increase the reliability of the data transmission, especially for vehicles at the cell edge~\cite{ISTM-2016}. In addition, device-to-device~(D2D) technology can be utilized for the direct V2V communication between multiple vehicles in proximity \cite{SNHH-2015,3GPP-REL14}.
Moreover, non-orthogonal multiple access (NOMA) schemes are considered in V2X systems to support more vehicles over limited spectrum resources as well as improving the spectrum efficiency~\cite{BLY-2017,K-2016}.


\par Different from aforementioned schemes, we aim to support more cellular V2X users by leveraging the unlicensed spectrum. Over the unlicensed spectrum, direct and distributed data transmission between nearby vehicles can be supported, forming the vehicular ad-hoc network~(VANET)\footnote{For example, IEEE 802.11 standards, including IEEE 802.11p and IEEE 802.11n, can be utilized for the data transmission over unlicensed spectrum~\cite{3GPP-REL14,BSCS-2013}.}.
{It worth noting that both cellular V2X and VANET can be used for safety-oriented applications. The differences between cellular V2V users and VANET users over the unlicensed spectrum are the modes of resource allocation and data structure.
1) The resources utilized by cellular V2V users are allocated by the BS in a centralized way~\cite{WSMK-2011,JYMJ-2016}, while VANET users access the channel in a distributed way over the unlicensed spectrum using the energy detection, i.e., each vehicle detects the power level of the channel and waits to access the channel until the detected power is lower than a threshold.
2) The data transmission of cellular V2V users follows the LTE standards, where the length of each data frame is constant and the transmission is terminated with the end of the frame. However, the data transmission of VANET users follows the 802.11 standards, i.e., the package size of each VANET user varies with the transmission demand of the VANET user. Moreover, the VANET user does not release the channel until the package has been transmitted completely, and thus, the transmission time of the package is uncertain.}
Hence, when the cellular V2X communication expands to unlicensed spectrum, it is necessary to consider the coexistence of cellular V2X users and VANET users in order to increase the number of active cellular V2X users as well as the reliability of data transmission.

\par However, new challenges will be brought by the coexistence system of cellular V2X and VANET:
(1) A suitable coexistence mechanism is needed over the unlicensed spectrum for cellular V2X users to offload the transmission demand.
Integrated allocation of dedicated cellular and unlicensed spectrum needs to be mitigated to take full advantage of the global control of cellular system.
Meanwhile, the interference to VANET users caused by cellular V2X users needs to be considered.
(2) Due to the large amount of data transmitted by numerous vehicles, dynamic resource scheduling in each time slot results in inevitable control overhead.
(3) The mobility of vehicles brings about the frequent change of the topological structure of vehicles, which further influences the channel conditions between vehicles. The change of channel conditions is likely to cause the failure of data transmission.

\par To tackle the above challenges, we propose an energy sensing based coexistence scheme over the unlicensed spectrum for both cellular V2X users and VANET users, in which cellular V2X users can share the open spectrum fairly with VANET users according to the sensed channel conditions.
To decrease the overhead, the semi-persistent scheduling (SPS) \cite{SPS-2015} method is considered for the resource allocation.
Moreover, We take the velocity of vehicles into the consideration to measure the change of channel conditions.
Furthermore, we formulate the time-frequency resource allocation problem. We allocate the dedicated cellular and shared subchannels, and schedule the time slots during the scheduling cycle, aiming to maximize the number of active cellular V2X users as well as reducing the interference to VANET users. Since the allocation of subchannels to cellular V2X users in different time slots can be considered as a matching between vehicles and time-frequency resources, we reformulate the problem as a two-sided many-to-many matching problem with peer effects \cite{M-2013,DBSL-2015,ZDSL-2016,ZLS-2016,BLCHW-2011}, which can be addressed by the proposed dynamic vehicle-resource matching algorithm (DV-RMA). 

\par In literature, several existing works also considered the coexistence mechanisms for the heterogeneous system. In \cite{LJMC-2013}, a cluster based heterogeneous systems was proposed for vehicular intersection collision avoidance service, which reduced the traffic amount between vehicles and BS as well as mitigating the overload in high density conditions. 
In \cite{SLCH-2016}, a hybrid overlay protocol was proposed to help vehicles to adaptively select the optimal communication mode, i.e. , ad-hoc communication or cellular communication, according to the transmission demand and network conditions. In~\cite{CYSM-2016},~\cite{3GPP-2013}, the coordination of LTE and Wi-Fi system was investigated for the cellular data offloading on the unlicensed spectrum. However, the resource allocation mechanism of all cellular V2X users over the unlicensed spectrum has not been considered. Moreover, the tradeoff between massive connectivity and user fairness has not been fully discussed.

\par The main contributions of our work in this paper are listed as follows.
\begin{itemize}
\item We propose an energy sensing based spectrum sharing scheme for cellular V2X users to share the unlicensed spectrum fairly with VANET users. The number of active cellular V2X users is maximized while the interference to VANET users is considered.
\item We quantify the interference to VANET users caused by cellular V2X users for utilizing the unlicensed spectrum by constructing the vehicle interference model, in which the interference is approximated as the interference area brought by the cellular V2X user.
\item We investigate the time-frequency resource allocation problem using matching theory. The DV-RMA algorithm in consideration of the peer effects is then designed. Besides, its stability, convergence and complexity are analyzed.
\end{itemize}

\par The rest of the paper is organized as follows. In Section \ref{sec:model}, the system model of spectrum sharing cellular V2X system is described. In Section \ref{sec:scheme}, we design the energy sensing based scheme for cellular V2X users to share the common spectrum with VANET users fairly. We formulate the subchannel allocation and time scheduling problem, and reformulate the problem as a two-sided many-to-many matching problem in Section \ref{sec:formulation}. In Section \ref{sec:matching}, we propose the DV-RMA to solve the problem. The stability, convergence and complexity of DV-RMA as well as the interference to VANET users are analyzed in Section \ref{sec:performance}. Simulation results are presented in Section \ref{sec:simulation_results} and conclusions are drawn in Section \ref{sec:conclusion}.

\section{System Model \label{sec:model}}%
In this section, we consider the scenario where cellular V2X users coexist with VANET users over the unlicensed spectrum. The scenario is described first, and then the interference of cellular V2X users is analyzed.

\begin{figure}[h]
\small
\centering
\includegraphics[width=0.9\textwidth]{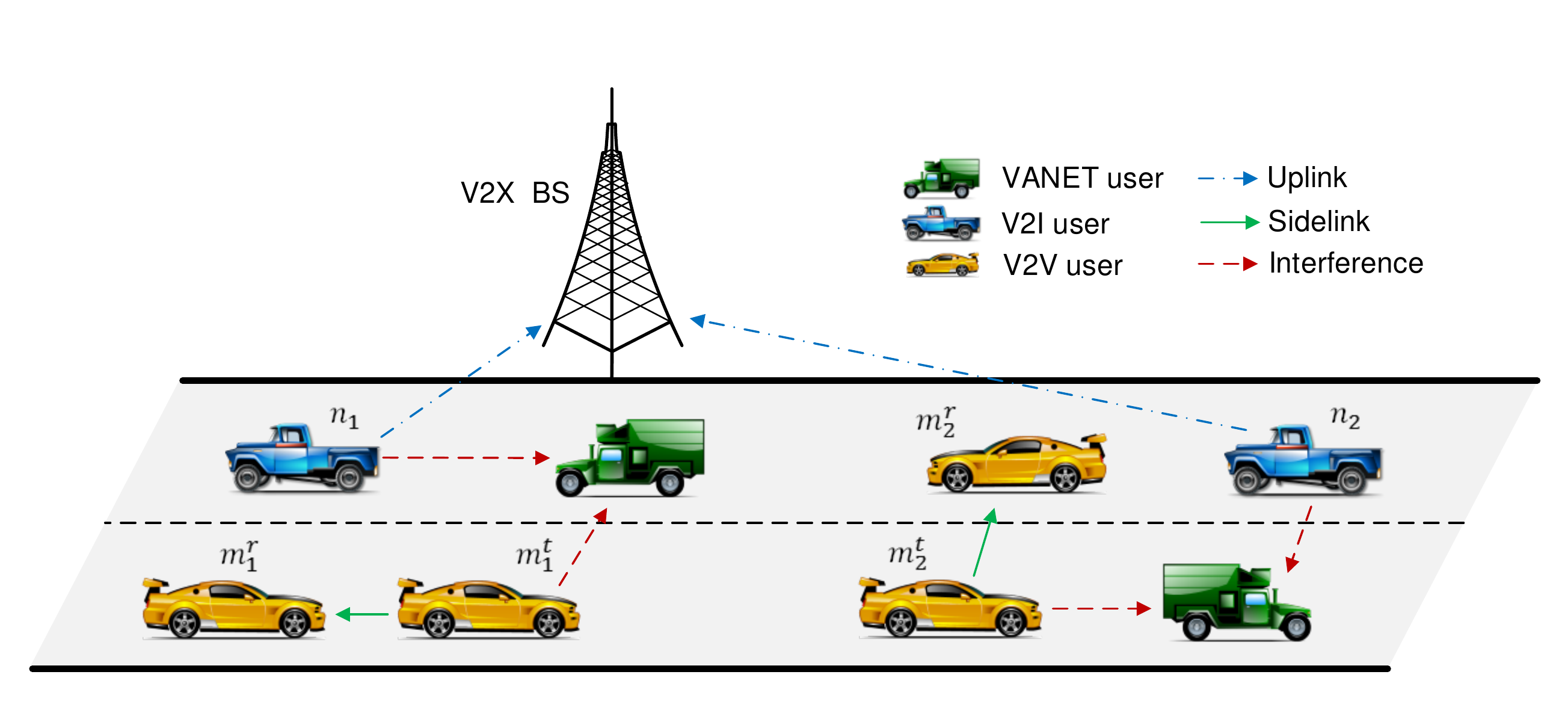}
\caption{System model of the coexistence system of cellular V2X and VANET.}
\label{Fig:system_model}
\end{figure}

\subsection{Scenario Description and Channel Model}
As shown in Fig.~\ref{Fig:system_model}, we consider a vehicular network which consists of $N$ V2I users, $M$ V2V users and VANET users. We denote the V2I transmitters by $n, n \in \mathcal{N} =\{1, \ldots, N\},$ and V2V users by $m, m \in \mathcal{M} =\{N+1, \ldots, N+M\}$, in which $m^t$ represents the transmitter and $m^r$ represents the receiver of V2V user $m$, respectively..

\par 
We assume that cellular V2X users access the dedicated cellular spectrum by orthogonally frequency domain multiplexing (OFDM).
The bandwidth of the dedicated cellular spectrum is divided into $K$ subchannels, denoted by $\mathcal{K}=\{1, 2,\ldots, K\}$.
Since the data transmission of cellular V2X users needs to follow the LTE standard in the unlicensed spectrum, the unlicensed spectrum can be divided into $K_u$ subchannels by the BS, denoted by $\mathcal{K}_u=\{K+1, K+2,\ldots, K+K_u\}$, to support multiple cellular V2X users simultaneously.

The timeline is divided into multiple subframes for cellular data transmission, each of which has the length of $T_s$. A scheduling frame consists of $T$ subframes, denoted by $\mathcal{T} = \{1,\ldots,T\}$. 
Particularly, vehicles are allocated dedicated cellular band in each subframe for control signal transmission to guarantee the reliability.

\par We assume that cellular V2X users transmit data with fixed power, denoted by $P^v$, on a dedicated cellular or unlicensed subchannel. In addition, let the minimum required received power be $P^r$.
The signal received by cellular V2X user $j$ from user $i$ can be expressed as
\begin{equation}
y_j = h_{i,j}x_i + n_j,
\label{e:signal_receiver}
\end{equation}
where $x_i$ is the signal sent by the cellular V2X $i$ and the noise $n_j$ follows independent Gaussian distribution with zero mean and the variance $\sigma^2$.

\par We adopt the free space propagation path-loss model with Rayleigh fading {\cite{Rayleigh}} to model the channel gain between cellular V2X users,
{ i.e.,  $P=P_0\cdot(d/d_0)^{-\alpha}\cdot |h_0|^2$, where $P_0$ and $P$ represent signal power measured at $d_0$ and $d$ away from the transmitter respectively, $\alpha$ is the path-loss exponent, and $h_0\sim \mathcal{CN}(0,1)$ is a complex Gaussian variable representing the Rayleigh fading. Moreover, we simplify the received power $P_0$ at $d_0 = 1$ equals the transmit power $P^v$.}
Since the mobility\footnote{{ The mobility of vehicles also brings about the problem of Doppler-shift. In our scenario, we focus on the urban environment with low velocities of $15\sim60 km/h$ \cite{3GPP-2016}, hence, the maximum Doppler-shift over the unlicensed spectrum ($2.4GHz$) is less than $0.27kHz$. }} of vehicles can severely affect the distance between them, the velocity information is utilized to estimate the distance between cellular V2X users $i$ and $j$. The received signal power of the user $j$ on one subchannel can be given by
\begin{equation}
p_{i,j}^{(r)}=P^v \cdot |h_{i,j}|^2,
\label{e:receiving_power}
\end{equation}
and the channel gain $h_{i,j}$ from the user $i$ to user $j$ can be expressed as
\begin{equation}
|h_{i,j}|^2=G \cdot |\textbf{d}_{i,j}+\textbf{v}_{i,j}\cdot t_w|^{-\alpha}\cdot |h_0|^2,
\label{e:channel_gain}
\end{equation}
 where $G$ is the constant power gain factor introduced by amplifier and antenna, $\textbf{d}_{i,j}$ is the distance vector from cellular V2X user $i$ to user $j$, $\textbf{v}_{i,j}$ is the relative velocity of cellular V2X user $i$ to user $j$, and $t_w$ is the waiting interval between the time point when the data to transmit is ready and the time point when the actual data transmission starts.

\subsection{Interference Analysis \label{sec:interference_analysis}}
We analyze the interference of cellular V2X users to VANET users and other cellular V2X users in this subsection. The VANET users and cellular V2X users occupy the channel in different ways, i.e., the VANET user occupies the whole channel, while cellular V2X users share the channel by dividing it into multiple subchannels.

\subsubsection{Interference to VANET \label{sec:wifi_interference_model}}
When the unlicensed channel is occupied by cellular V2X users, the performance of VANET users is affected. To measure the degradation of the performance of VANET users, we define the interference range of cellular V2X users to VANET users as the area where the received power from the cellular V2X users exceeds the predefined threshold $P^r$.
According to the free path-loss model, the interference range is a circle with the cellular V2X user at the center. The radius of interference range of the cellular V2X user $i$, denoted by $L_i$, is determined by the channel fading and the minimum required receiving power $P^r$, as shown below:
\begin{equation}
L_i = \log_{\alpha}(\frac{P^v G |h_0|^2}{P^r}).
\label{e:interference_radius}
\end{equation}

\par When a new cellular V2X user $i$ accesses the unlicensed spectrum, the total interference range is likely to increase. Only the additional interference range of the user $i$ is considered, which is calculated as the area of the user $i$'s interference circle excluding the overlapping part, i.e., the shaded part in Fig.~\ref{Fig:vehicle_interference}.

\begin{figure}
\small
\centering
\includegraphics[width=0.6\textwidth]{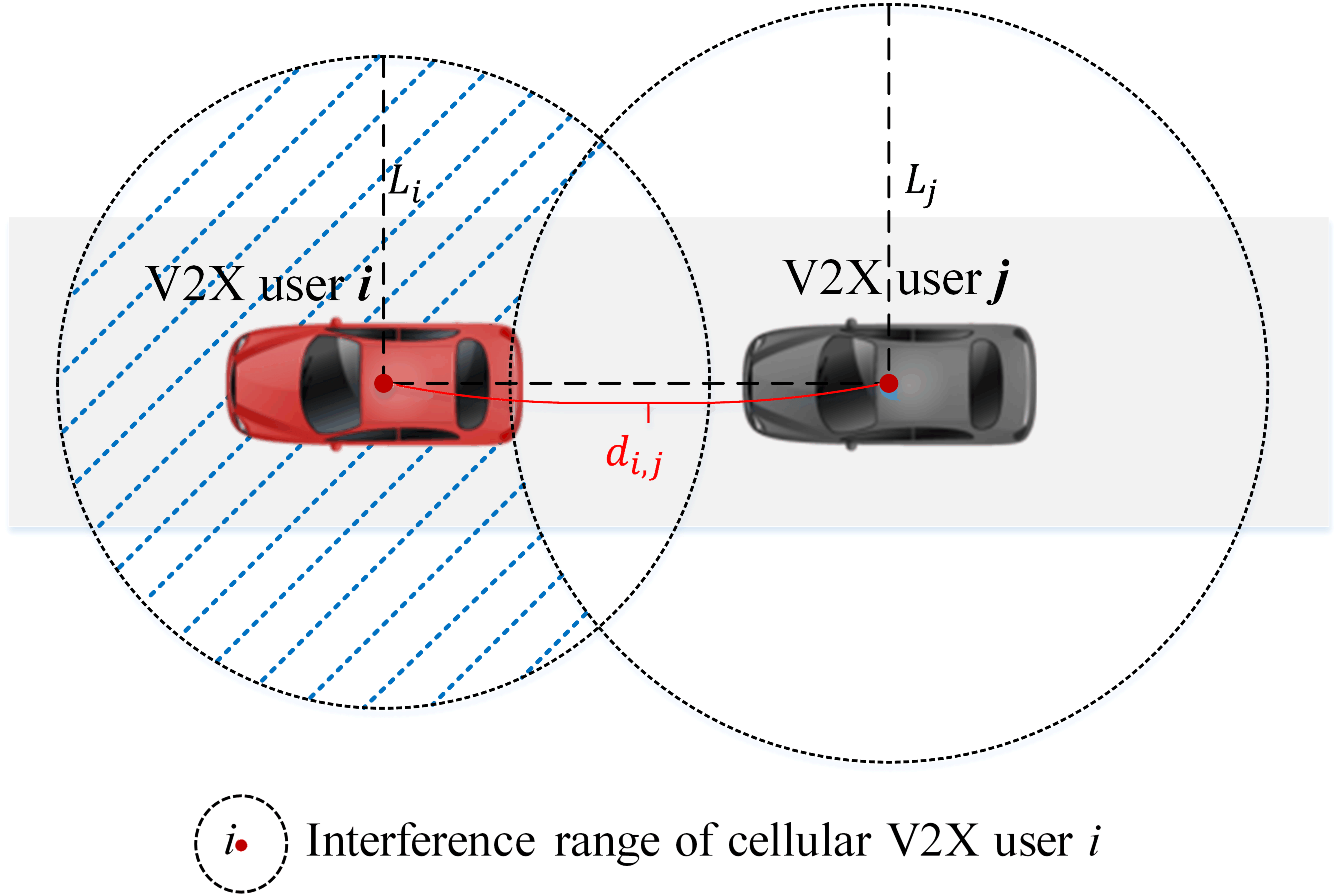}
\caption{Interference of the cellular V2X user to the VANET.}
\label{Fig:vehicle_interference}
\end{figure}

\par Let $d_{i,j}$ be the distance between cellular V2X users $i$ and $j$, and $\mathcal{V}$ be the set of cellular V2X users which occupy unlicensed subchannels. To better depict the size of the additional interference area brought by cellular V2X user $i$ given the existence of cellular V2X user $j$, we express the overlapping part in the following proposition.

\begin{prop}
Without loss of generality, we assume that $L_i<L_j$, then {the area of} the overlapping part between the interference area of cellular V2X users $i$ and $j$ is
\begin{equation}
O_{i,j}=
\left\{
\begin{array}{ll}
L_i^2 \arccos(\frac{ d_{i,j}^2+L_i^2-L_j^2}{2 L_i d_{i,j}}) + L_j^2 \arccos(\frac{ d_{i,j}^2-L_i^2+L_j^2}{2 L_j d_{i,j}}) - \\ \frac{1}{2}\sqrt{2L_j^2(L_i^2+d_{i,j}^2)-(L_i^2-d_{i,j}^2)^2-L_j^4}, &$if$ \quad d_{i,j}<\sqrt{L_j^2-L_i^2}, \\
\pi L_i^2 - L_i^2 \arccos(\frac{ L_j^2-d_{i,j}^2-L_i^2}{2 L_i d_{i,j}}) + L_j^2 \arccos(\frac{ d_{i,j}^2-L_i^2+L_j^2}{2 L_j d_{i,j}}) - \\ \frac{1}{2}\sqrt{2L_j^2(L_i^2+d_{i,j}^2)-(L_i^2-d_{i,j}^2)^2-L_j^4}, &$if$ \quad  d_{i,j}\geq \sqrt{L_j^2-L_i^2}.
\end{array}
\right.
\label{e:overlap}
\end{equation}
\end{prop}

\begin{proof}
See Appendix \ref{app:proof_prop}.
\end{proof}

{Assuming that the additional interference area brought by cellular V2X user $i$ to user $j$, denoted by $f_{i,j}$, is the difference between the interference area of user $i$ and the overlapping interference area of users $i$ and $j$.  Hence, the additional interference area is shown as below: }
\begin{equation}
f_{i,j}=
\left\{
\begin{array}{ll}
\pi L_i^2 -O_{i,j}, &$if$ \quad d_{i,j}<L_i+L_j, \\
\pi L_i^2, &$if$ \quad  d_{i,j}\geq L_i+L_j.
\end{array}
\right.
\label{e:interference_weight}
\end{equation}
 Given cellular V2X user $j$ whose interference range has been considered, the additional interference area $f_{i,j}$ of cellular V2X user $i$ varies with user $j$, determined by the the interference radius $L_j$ and the distance between cellular V2X users $i$ and $j$. { Considering only the strongest interference,} the additional interference of cellular V2X user $i$, denoted by $f_i$, is defined as shown bellow:
\begin{equation}
f_{i}=\min\limits_{j\in \mathcal{V}}f_{i,j}.
\label{e:min_weight}
\end{equation}

\subsubsection{Interference Analysis of Cellular V2X System}
We first define an indicator $\theta_{i,k}^{(t)}$ to describe subchannel utilization, where
\begin{equation}
\theta_{i,k}^{(t)}=
\left\{
\begin{array}{ll}
1, &$if subchannel $k$ is allocated to cellular V2X user $i$ in subframe $t$, $ \\
0, &$otherwise$.
\end{array}
\right.
\label{e:V2I_indicator}
\end{equation}

 We assume that in a subframe one subchannel can be allocated to at most one V2I user, and multiple V2V users can share the same subchannel. Moreover, since cellular V2X users and VANET users do not utilize the same channel simultaneously, the interference from VANET users to cellular V2X users does not need to be considered.
Therefore, the SINR at the receiver of BS from V2I user $n$ over subchannel $k$ in subframe $t$ can be expressed as
\begin{equation}
\gamma_{n,k}^{(t)}=\frac{\theta_{n,k}^{(t)} P^v|h_{n,0}|^2}{\sigma^2+\displaystyle{\sum_{i=N+1}^{N+M}}{\theta_{i,k}^{(t)} P^v|h_{i,0}|^2}},
\label{e:SINR_licensedV2I}
\end{equation}
where $h_{i,0}$ represents the channel gain from cellular V2X user $i$ to BS.

The SINR at the receiver $m^r$ from V2V transmitter $m^t$ over subchannel $k$ in subframe $t$ is denoted by
\begin{equation}
\gamma_{m,k}^{(t)}=\frac{\theta_{m,k}^{(t)} P^v|h_{m,m}|^2}{\sigma^2+\displaystyle{\sum_{i=1, i\neq m}^{N+M}}{\theta_{i,k}^{(t)} P^v|h_{i,m}|^2}}.
\label{e:SINR_unlicensedV2V}
\end{equation}
The data rate of V2I user $n$ and V2V user $m$ over subchannel $k$ in subframe $t$ is respectively expressed as
\begin{equation}
R_{n,k}^{(t)}=\log_2(1+\gamma_{n,k}^{(t)}),
\label{e:rate_licensed}
\end{equation}
and
\begin{equation}
R_{m,k}^{(t)}=\log_2(1+\gamma_{m,k}^{(t)}).
\end{equation}

\section{Energy Sensing based Spectrum Sharing {Scheme} \label{sec:scheme}}

In this section, we propose the energy sensing based spectrum sharing (ESSS) {scheme} for cellular V2X users to share the unlicensed spectrum fairly with VANET users. 
Different from the users in a centralized downlink system, VANET users need to sense the channel conditions individually for their own data transmission, requiring a specific mechanism to reduce the collision of data transmission. In addition, many safety-critical applications have raised stringent latency requirements on the vehicular network, and thus, the sensing duration needs to be limited to a tolerable level. Moreover, compared with the traditional ad-hoc vehicular network, new challenges have been posed due to the existence of V2I communications.

\par In this scheme, the adaptive duty cycle \cite{QUAL-2014,BAH-2014} is adopted, which represents a constant period of data transmission for cellular V2X users and VANET users. As shown in Fig.~\ref{Fig:scheme_design}, each duty cycle is divided into the sensing period and adaptive transmission period.


\begin{figure}[t]
\small
\centering
\includegraphics[width=0.95\textwidth]{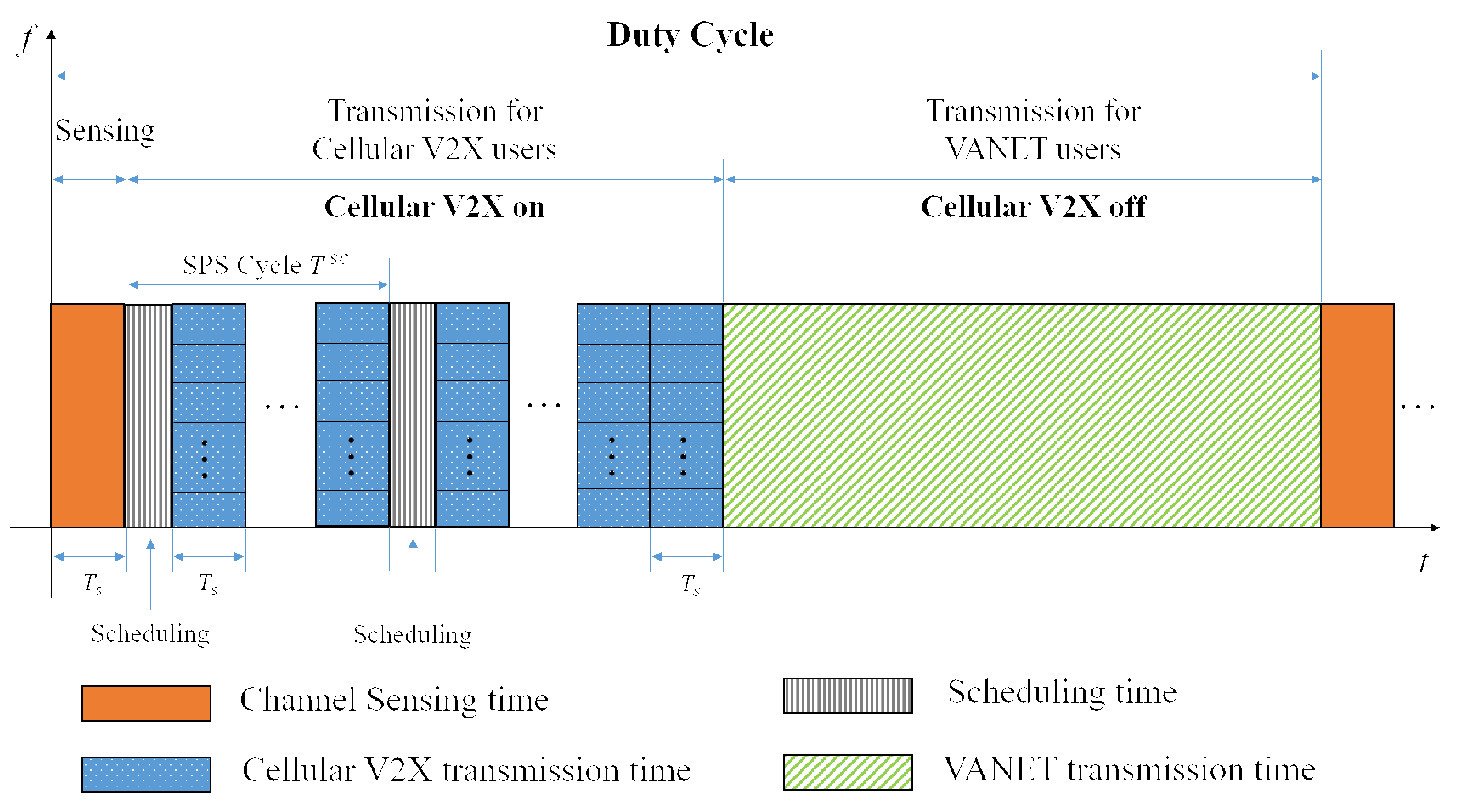}
\caption{Structure of the duty cycle in the ESSS scheme.}
\label{Fig:scheme_design}
\end{figure}

\subsection{Channel Sensing}
In the sensing period, cellular V2X users sense the channels to select one of the unlicensed channels and occupy it. 
The cellular V2X users first perform sensing to select the most suitable unlicensed channel for transmission. They measure the interference level over each unlicensed channel by energy detection. The channel is considered to be idle if the interference level is lower than the sensing threshold. If there is an idle channel, the cellular V2X users will occupy this unlicensed channel, otherwise, they need to share the channel of the least interference level with VANET users.
When the channel is shared, in order to reduce the collision of data transmission between cellular V2X users and VANET users, cellular V2X users will keep sensing until the selected channel is idle and then occupy it immediately.

The channel selection mechanism monitors the operating channel status and keeps tracking the most suitable channel, i.e., once the interference level of the operating channel is beyond the threshold and there exists another channel over which the interference level is detected lower than that of the operating channel, the cellular V2X users will switch to the less interfered channel for data transmission.

\subsection{Sensing Based Adaptive Transmission}
\par The length of cellular V2X transmission is set based on the sensing results. When the unlicensed channel is sensed idle, the cellular V2X transmission period works throughout the duty cycle. Otherwise, the same duration as that of the cellular V2X transmission period is kept for the VANET transmission.

\par In the cellular V2X transmission period, SPS method is utilized to reduce the overhead and resource waste. Dedicated cellular and unlicensed subchannels are allocated to the users periodically and temporarily occupied during the SPS cycle. The cellular V2X transmission period includes several SPS cycles, denoted by $T^{sc}$, and one SPS cycle includes $T$ data transmission subframes, i.e., $T^{sc} = T \cdot T_s$.
At the beginning of each SPS cycle, the transmission demand of each user within this cycle is collected by the BS for scheduling. {The BS takes the demands of all cellular V2X users into the comprehensive consideration and allocates the time and frequency resources accordingly in a centralized way, representing the results that serve as the bound of the performance to be achieved by a distributed online scheme in such scenarios.}
Moreover, the position and velocity information of the vehicles is broadcasted periodically by themselves, and updated during the scheduling interval.
The data transmission of cellular V2X users follows the LTE standards and VANET users, however, cannot access the channel during these subframes.
\par In the VANET transmission period, VANET users can utilize the unlicensed spectrum freely following the IEEE 802.11 standards, while cellular V2X users can only transmit over the dedicated cellular spectrum. The next duty cycle will start once the VANET transmission period terminates.

\section{Problem Formulation \label{sec:formulation}}
In this section, we focus on the dedicated cellular and unlicensed subchannel allocation as well as time scheduling during the SPS cycle, which can be formulated as the subchannel allocation and time scheduling problem. We then reformulate this problem as a two-sided many-to-many matching problem between vehicles and time-frequency resources.
\subsection{Subchannel Allocation and Time Scheduling Problem}
Define a binary transmission matrix $\bm{\beta}$ to indicate whether cellular V2X user $i$ can perform a successful transmission. The indicator element over subchannel $k$ in subframe $t$ can be expressed as
\begin{equation}
\beta_{i,k}^{(t)}=
\left\{
\begin{array}{ll}
1, &$if$ \ \gamma_{i,k}^{(t)} \geq \gamma_{th}, \\
0, &$if$ \ \gamma_{i,k}^{(t)} < \gamma_{th},
\end{array}
\right.
\end{equation}
where $\gamma_{th}$ is the SINR threshold for successful detecting.

According to the interference measurement model, the total interference range to VANET users in subframe $t$ is the sum of interference range of each cellular V2X user that occupies unlicensed subchannels. We then introduce the penalty term to reflect the total interference to VANET users caused by cellular V2X users in subframe $t$:
\begin{equation}
C_i^{(t)}=f_i \cdot \sgn(\sum_{k\in\mathcal{K}_u}\theta_{i,k}^{(t)}),
\label{e:penalty_term}
\end{equation}
in which $f_i$ is the additional interference of cellular V2X user $i$ presented in (\ref{e:interference_weight}).

\par To satisfy the reliability requirements of multiple cellular V2X users, we aim to maximize the number of active cellular V2X users including V2I and V2V users, while constraining the interference to VANET users to a tolerable level during the SPS cycle. Therefore, the scheduling problem is formulated as follows:
\begin{align}
\max_{\theta_{i,k}^{(t)}} & \left( \sum_{t=1}^{T}\sum_{i=1}^{N+M}\sum_{k=1}^{K+K_u}\beta_{i,k}^{(t)}\theta_{i,k}^{(t)} - \lambda \cdot \sum_{t=1}^{T}\sum_{i=1}^{N+M} C_i^t(\theta_{i,k}^{(t)}) \right), \label{e:objective_function}\\
s.t. & \quad \sum_{i=1}^{N}\theta_{i,k}^{(t)} \leq 1, \ \forall k\in\mathcal{K}\cup\mathcal{K}_u, 1\leq t \leq T, \tag{\ref{e:objective_function}a}           \label{e:objective_V2I_licensed_num_limit}\\
      & \quad \sum_{k=1}^{K+K_u}\sum_{t=1}^{T}\theta_{i,k}^{(t)} \leq S, \ \forall i\in\mathcal{N}\cup\mathcal{M}, \tag{\ref{e:objective_function}b}             \label{e:objective_user_subchannels_limit}\\
      & \quad \sum_{i=1}^{N+M}\theta_{i,k}^{(t)}\leq Q, \ \forall k\in\mathcal{K}\cup\mathcal{K}_u, 1\leq t \leq T, \tag{\ref{e:objective_function}c}             \label{e:objective_subchannel_users_limit} \\
       & \quad \sum_{k=1}^{K}\theta_{i,k}^{(t)} \geq 1, \ \forall i\in\mathcal{N}\cup\mathcal{M}, 1\leq t \leq T, \tag{\ref{e:objective_function}d}            \label{e:objective_control_channel_limit}
\end{align}
where $\lambda\geq 0$ is the penalty factor. Constraint (\ref{e:objective_V2I_licensed_num_limit}) shows that one subchannel can be allocated to no more than one V2I user in one subframe;
 (\ref{e:objective_user_subchannels_limit}) shows that one vehicle can occupy no more than $S$ time-frequency resources during one SPS cycle;
(\ref{e:objective_subchannel_users_limit}) reflects that one subchannel can be occupied by no more than $Q$ cellular V2X users in one subframe;
(\ref{e:objective_control_channel_limit}) shows that at least one dedicated cellular subchannel needs to be allocated to each vehicle in each subframe for the reliable transmission of control signals.

\subsection{Matching Problem Formulation}
\par The time resources and frequency resources are orthogonal, and thus, we take the time-frequency resource as a whole, and define the subchannel $k$ in the subframe $t$ as $W_{k,t}$, where $k\in\mathcal{K}\cup\mathcal{K}_u$ and $t\in \mathcal{T}$. We denote the cellular V2X user set by $\mathcal{U} = \mathcal{N}\cup\mathcal{M}$ and time-frequency resource set as $\mathcal{W}=(\mathcal{K}\cup\mathcal{K}_u)\times\mathcal{T}$.
$\mathcal{U}$ and $\mathcal{W}$ are considered as two disjoint sets of selfish and rational players.
The time-frequency resource allocation to cellular V2X users can then be considered as a matching between cellular V2X user set $\mathcal{U}$ and time-frequency resource set $\mathcal{W}$.
Hence, the time-frequency allocation problem in (\ref{e:objective_function}) can be formulated as a two-sided many-to-many matching problem.
For convenience, we denote $(V_i, W_{k,t})$ as a \emph{matching pair} if $W_{k,t}$ is allocated to vehicle user $V_i$. 

\par In order to describe the interest of each player, we assume that each player has a \emph{preference} over the subset of the opposite side. For vehicle $V_i$, its preference is determined by the data rate over the time-frequency resource.
If vehicle $V_i$ prefers $W_j$ to $W_{j'}$, then we have
\begin{align}
\label{e:preference_v}
&W_{k,t} \succ_{V_i} W_{k',t'} \ \Leftrightarrow \ R_{i,k}^{(t)} > R_{i,k'}^{(t')}, \\
&\forall k\in\mathcal{K}\cup\mathcal{K}_u, t\in\mathcal{T}, |k-k'|+|t-t'|\neq 0, \nonumber
\end{align}
which implies that vehicle $V_i$ obtains higher data rate over resource $W_{k,t}$ than resource $W_{k',t'}$.

The preference of time-frequency resource $W_{k,t}$ is determined by the the utility and total interference area.
To be specific, the utility of resource $W_{k,t}$ is the number of the active cellular V2X users over it, which can be expressed as 
\begin{equation}
\label{e:preference_of_W}
U_{W_{k,t}}(\bm{\Theta}) = \sum_{i=1}^{N+M}\beta_{i,k}^{(t)}\theta_{i,k}^{(t)},
\end{equation}
where $\bm{\Theta}=(\theta_{1,1}^{(1)},\dots, \theta_{N+M,K+K_u}^{(T)})$ be the resource allocation indicator.
The total interference area of cellular V2X users during the SPS cycle can be expressed as
\begin{equation}
C(\bm{\Theta}) = \sum_{t=1}^{T}\sum_{i=1}^{N+M}C_i^t.
\end{equation}
Therefore, resource $W_{k,t}$ preferring vehicle $V_i$ to vehicle $V_{i'}$ is given by
\begin{align}
\label{e:preference_w}
& V_i \succ_{W_{k,t}} V_{i'} \\\nonumber
\Leftrightarrow \ & U_{W_{k,t}}(\bm{\Theta})-\lambda C(\bm{\Theta})\ > \ U_{W_{k,t}}(\bm{\Theta}')-\lambda C(\bm{\Theta}').\nonumber
\end{align}
With the above preference relations, we then define the preference list.
\begin{myDef}
A preference list is an ordered set containing all the possible subsets of the opposite set for player $i$ $(i\in\mathcal{U}\cup\mathcal{W})$. In consideration of $q$ subsets of the opposite set of player $i$, denoted by $A_1, A_2,\dots,A_q$, the preference list of player $i$ is defined as $P(i)=\{A_1, A_2,\dots, A_q\}$, which represents that $A_1, A_2,\dots,A_q$ are possible to be the player $i$'s matching pair and $A_1\succ_{i}A_2\succ_{i}\dots \succ_{i}A_{q}$.
\end{myDef}

The preference list can be determined by the BS at the beginning of each SPS cycle based on the channel state information. Moreover, the preference of the player, e.g. $i$, is \emph{transitive}, implying that if $A\succ_{i}A'$ and $A'\succ_{i}A''$, then $A\succ_{i}A''$, where $A, A'$ and $A''$ are subsets of the opposite set of player $i$.

\par With the above definitions of matching pair and preference list, we then formulate the optimization problem shown in (\ref{e:objective_function}) as a \textbf{two-sided many-to-many matching problem}.

\par Given two disjoint sets, i.e., cellular V2X user set $\mathcal{U}=\{V_1,V_2,\dots,V_{N+M}\}$ and time-frequency resource set  $\mathcal{W}=\{W_{1,1},\dots, W_{K+K_u,T}\}$, a many-to-many matching $\Psi$ is a mapping from the set $\mathcal{U}\cup\mathcal{W}$ to the subsets of $\mathcal{U}\cup\mathcal{W}$:
\begin{enumerate}[1)]
\item $\Psi(V_i)\in\mathcal{W}, \forall \ V_i \in\mathcal{U}$,  and $\Psi(W_{k,t})\in\mathcal{U}, \forall \ W_{k,t}\in \mathcal{W}$;
\item $\Psi(V_i) = W_{k,t} \Leftrightarrow \Psi(W_{k,t}) = V_i$;
\item $|\Psi(W_{k,t}) \cap \mathcal{N}|\leq 1$;
\item $|\Psi(V_i)|\leq S$;
\item $|\Psi(W_{k,t})|\leq Q$;
\item If $\exists W_{k,t} \in \Psi(V_i)$ in subframe $t$, there exists $W_{k',t}\in (\mathcal{K}\times\mathcal{T}), W_{k',t} \in \Psi(V_i)$.
\end{enumerate}
Conditions 1) and 2) imply that the cellular V2X user and the time-frequency resource are matched with each other manually. Conditions 3)-5) correspond to constraint (\ref{e:objective_V2I_licensed_num_limit})-(\ref{e:objective_subchannel_users_limit}), respectively. Condition 6) means that a minimum of one dedicated cellular time-frequency resource is allocated to $V_i$ in each occupied subframe according to constraint (\ref{e:objective_control_channel_limit}).

\par Compared to the traditional matching problems, our formulated one is rather complicated due to the following reasons.
First, \emph{externalities} or \emph{peer effects} [\ref{ref:matching_peer_effect}] are considered because the preference of the vehicle is influenced by the other vehicles' co-channel interference over the same time-frequency resource, making the decisions of all vehicles are correlative with the others' behaviors.
Second, the mobility of vehicles makes the topology time-varying, and thus we need to consider the change of vehicles' positions in each subframe during the SPS cycle.
Third, the density of vehicle users is high, and each player can match with any subset of the opposite set, resulting in numerous matching combinations when the number of vehicles is large.
For these reasons, existing algorithms \cite{RS-1992} cannot be directly applied to solve our problem. Thus, we develop a dynamic matching algorithm. 

\section{Dynamic vehicle-resource Matching Algorithm \label{sec:matching}}%

In this section, we design a dynamic vehicle-resource matching algorithm~(DV-RMA). First we illustrate the matching process given the static preference list in our proposed algorithm, then we elaborate on the whole process of the dynamic algorithm in consideration of the peer effects.

\subsection{Algorithm Design}\label{sec:Algorithm Design}

The DV-RMA is constructed based on the iteratively renewed preference lists, which are updated at the beginning of each static matching process.
The static matching process consists of several rounds of vehicles' proposal and time-frequency resources' acceptation/rejection. Three steps of each round are described in detail as follows.

\subsubsection{Vehicles propose to the time-frequency resources}
 Each \emph{incompletely matched}\footnote{A cellular V2X user is incompletely matched if it is matched with fewer than $S$ time-frequency resources.} vehicle proposes itself to its most preferred time-frequency resource according to its preference list. Vehicles then remove the time-frequency resource they have proposed to from their preference lists in the current matching process, and wait for the responses of the time-frequency resources. The peer effects brought by other vehicles are unpredictable for each vehicle. Therefore, the proposal in each round of matching process is based on the original preference list of vehicles constructed at the beginning of the current matching process.

\subsubsection{The time-frequency resources accept or reject the proposal}
  When vehicles propose to the resource, the time-frequency resource needs to evaluate the utility it obtains from accepting the proposal against the loss brought by the increase of the co-channel interference. We introduce the concept of \emph{blocking pair} to describe the mechanism of the time-frequency resource deciding whether to accept the vehicle's proposal.

\begin{myDef}
    Provided a matching $\Psi$ and a pair $(V_i, W_{k,t})$ with $V_i \notin \Psi(W_{k,t})$ and $W_{k,t}\notin\Psi(V_i)$. $(V_i, W_{k,t})$ is a \textbf{blocking pair} if
    $(1) W_{k,t}\in P(V_i);$
     $(2)\{V_i\}\cup\Psi(W_{k,t})\succ_{W_{k,t}}\Psi(W_{k,t});$
     $(3)\{W_{k,t}\}\cup\Psi(V_i) \succ_{V_i}\Psi(V_i).$
\end{myDef}

\par The blocking pair $(V_i, W_{k,t})$ is considered only when $V_i$ and $W_{k,t}$ have not matched with each other and $W_{k,t}$ is in the preference list of $V_i$. Moreover, the match of $V_i$ and $W_{k,t}$ can increase the interest of both $V_i$ and $W_{k,t}$, representing they both prefer to match with each other and the current matching structure is not optimal.

\par With the definition of blocking pair, we then describe how the time-frequency resource $W_{k,t}$ decides whether to accept the proposal of $V_i$. The proposals of V2I users and V2V users are discussed sequentially.
Firstly we discuss the \emph{V2I users' proposals\label{par:V2I_proposal}}. Since one time-frequency resource can be matched with no more than one V2I user, $W_{k,t}$ need to check whether the proposal is from the V2I user.
If $W_{k,t}$ has kept the matching proposal of a V2I user $V_{*}$, $W_{k,t}$ can match with only one of $V_{*}$ and $V_i$ to maximize its own utility as well as reducing the total interference area to VANET users.
Otherwise, $W_{k,t}$ makes different decisions in the following three cases:
\begin{itemize}
\item Accept when $W_{k,t}$ is unsaturated\footnote{A time-frequency resource is unsaturated if it is matched with fewer than $Q$ vehicles.} and $(V_i,W_{k,t})$ forms the blocking pair.
\item Reject when $W_{k,t}$ is unsaturated and $(V_i,W_{k,t})$ is not the blocking pair.
\item Reject the \emph{least contributed\footnote{Excluding the interference range, the utility that the least contributed vehicle brings is least.}} vehicle $V_m\in \Psi(W_{k,t})\cup V_i$ when $W_{k,t}$ is saturated\footnote{A time-frequency resource is saturated if it is matched with $Q$ vehicles.}.
\end{itemize}
Secondly, we discuss the \emph{V2V users' proposals}. The process is the same as that of V2I users' proposal when the resource doesn't match with any V2I user, three cases of which has been mentioned before.
\subsubsection{Termination conditions of the matching process}
The matching process terminates when no vehicle users propose themselves to time-frequency resources anymore, which means that either the vehicle has been matched with $S$ time-frequency resources or the preference list of the vehicle is empty.

\subsubsection{Dynamic algorithm}
\par The dynamic algorithm consists of multiple static matching processes. We first discuss the update of preference list between two matching processes, and then introduce the \emph{incompatible list} to avoid the repeated proposals which has been rejected in the same matching structure before.

\par As the result of \emph{peer effect}, the preference lists of vehicles are correlated with the matching structure of vehicles and time-frequency resources. Hence, the static preference list is no longer suitable to our matching problem. We dynamically adjust the preference lists of vehicles according to the matching result of each static \emph{matching process}.
The preference of the vehicle $V_i$ over each time-frequency resource $W_{k,t}\in\mathcal{W}$ is determined by the data rate $R_{i,k}^{(t)}$ based on the current matching pairs over $W_{k,t}$. To increase the possibility of the transmitted data being successfully received, the vehicle only proposes to access the time-frequency resources over which the SINR is larger than the SINR threshold $\gamma_{th}$. Hence, the time-frequency resources over which $\gamma_{i,k}^{(t)}<\gamma_{th}$ are not in the preference list of the vehicle.
The \emph{matching process} will be repeated until all the vehicles are either completely matched or have an empty preference list in the first round of matching process.

\par In order to reduce the complexity of our algorithm, repeated attempts of formerly rejected matching proposals need to be avoided. For vehicle $V_i$, some time-frequency resources do not accept $V_i$'s matching proposal based on current matching condition, since they have rejected vehicle $V_i$ in the previous iterations of matching under the same matching structure. To exclude these time-frequency resources from the preference list of vehicles, we introduce the concept of \emph{incompatible pair} defined as follows.

\begin{myDef}
For $V_i\in\mathcal{U}, W_{k,t}\in\mathcal{W},$ if $\Psi(W_{k,t}) \succ_{W_{k,t}}\{V_i\}\cup\Psi(W_{k,t})$, $V_i \notin \Psi(W_{k,t})$ and $W_{k,t}\notin\Psi(V_i)$, $\Psi(W_{k,t})$ is in \textbf{incompatible list} for $V_i$ over $W_{k,t}$, denoted by $\Psi(W_{k,t})\in F_{V_i}(W_{k,t})$.
\end{myDef}

The time-frequency resource makes the same decision when facing the same matching proposal under the same matching structure. Hence, the vehicle users will not propose to the time-frequency resource with the matching structure recorded in the \emph{incompatible list} to avoid failed proposals. The matching pairs of $W_{k,t}$ are added into the incompatible list of vehicle $V_i$ once $W_{k,t}$ rejects the proposal of vehicle $V_i$.

\begin{algorithm}
\small
\label{DV-RMA_algorithm}
\caption{Dynamic vehicle-resource matching algorithm} 
\hspace*{0.02in} {\bf Input:} 
Set of vehicles $\mathcal{U}$ and time-frequency resources $\mathcal{W}$.\\
\hspace*{0.02in} {\bf Output:} 
A two-sided many-to-many matching $\Psi$.\\
\\
\hspace*{0.02in} {\bf Phase1. Initialization Phase:}
\begin{algorithmic}
\State Build a matching $\Psi$ randomly satisfying all the constraints; 
\State Initialize the incompatible list $F_{V_i}(W_{k,t})$ as $\emptyset, \forall V_i\in\mathcal{U},W_{k,t}\in\mathcal{W}$; 
\State Set a binary variable $\mu$ representing whether any vehicle proposes to resources;
\State Set a variable $\kappa$ representing the iteration times of one matching process;
\end{algorithmic}

\hspace*{0.02in} {\bf Phase2. Matching Phase:}
\begin{algorithmic}
\Loop
    \State  Construct preference list $P(V_i)$ based on (\ref{e:rate_licensed}), (\ref{e:preference_v});
    \State  Delete $W_{k,t}$ from $P(V_i)$ if $\beta_{i,k}^{(t)}=0$;
    \State  Delete $W_{k,t}$ from $P(V_i)$ if $\Psi(W_{k,t})\in F_{V_i}(W_{k,t})$;
    \If{ $P(V_i)=\emptyset, \forall V_i\in\mathcal{U}$}
        \State $\mu=0$;
    \Else
        \State $\mu=1$;
    \EndIf
    \State $\kappa = 0$;
    \Loop
        \State $\kappa=\kappa+1$;
        \State $V_i$ proposes to $P(V_i)[1], \forall V_i\in \mathcal{U}$;
        \If{$\mu==0$} 
            \State Terminate the static matching process.
            \State \textbf{break};
        \EndIf
        \ForAll{$W_{k,t}$ that vehicles propose to}
            \If{$V_i$ is a V2I user}
                \If{$W_{k,t}$ has matched with V2I user}
                    \State $W_{k,t}$ chooses the V2I user that it prefers based on (\ref{e:preference_w}).
                \Else
                    \If{ $|\Psi(W_{k,t})|<Q$ \textbf{and}  $(V_i,W_{k,t})$ is a blocking pair}
                        \State $W_{k,t}$ accepts $V_i$;
                    \Else
                        \State $W_{k,t}$ rejects $V_i$;
                    \EndIf
                \EndIf
            \Else{ $ V_i$ is a V2V user}
                \If{ $|\Psi(W_{k,t})|<Q$ \textbf{and} $(V_i,W_{k,t})$ is a blocking pair}
                    \State $W_{k,t}$ accepts $V_i$;
                \Else
                    \State $W_{k,t}$ rejects $V_i$;
                \EndIf
            \EndIf
            \State Add $\Psi(W_{k,t})$ into the rejected vehicle's incompatible list;
        \EndFor
    \EndLoop
    \If{$\kappa ==1$}
        \State Terminate the Matching Phase.
        \State \textbf{break};
    \EndIf
\EndLoop
\end{algorithmic}

\hspace*{0.02in} {\bf Matching finished}
\end{algorithm}

\subsection{Algorithm Description}
As shown in \textbf{Algorithm 1}, the DV-RMA consists of the initialization phase and matching phase.
In the \emph{initialization phase}, we consider a random matching between vehicle users and time-frequency resources, with all the constraints satisfied. We initialize the incompatible lists of all vehicle users as $\emptyset$.
In the \emph{matching phase}, the process of each matching round is consist of updating the preference list and looking for stable matching based on current preference list.

\subsubsection{Update the preference list}
The preference lists of vehicles are updated according to the matching structure. The time-frequency resource is deleted if $\gamma_{i,k}^{(t)}<\gamma_{th}$ or the current matching structure $\Psi(W_{k,t})$ has been recorded in the incompatible list.

\subsubsection{Matching process}
 We obtain the matching result based on current preference list, as described in Section~\ref{sec:Algorithm Design}. Moreover, the vehicle $V_i$ rejected by the time-frequency resource $W_{k,t}$ will add the matching pairs of $W_{k,t}$, denoted by $\Psi(W_{k,t})$, into its incompatible list, denoted by $F_{V_i}(W_{k,t})$.

\subsubsection{Termination conditions of the matching phase}
 The matching phase terminates when no vehicle users propose in the first round of \emph{matching process}. Finally, a two-sided many-to-many matching is returned as the result of our algorithm.

\section{Performance Analysis} \label{sec:performance}
In this section, we first analyze the stability, convergence and complexity of the proposed algorithm DV-RMA. We then discuss the influence of interference caused by cellular V2X users on the VANET users.

\subsection{Stability, Convergence and Complexity of the Algorithm} \label{sec:complexity}
\subsubsection{Stability and Convergence}
We first introduce the concept of \emph{pairwise stable} matching, and then prove that the DV-RMA converges to a \emph{pairwise stable} matching.
\begin{myDef}
A matching $\Psi$ is defined as \textbf{pairwise stable} if it is not blocked by any pair which does not exist in $\Psi$.
\end{myDef}
Based on the definition of pairwise stable, we analyze the stability and convergence of the DV-RMA.
\begin{thm}
\label{thm:pairwise stable}
If the DV-RMA converges to a matching $\Psi$, then $\Psi$ is a \textbf{pairwise stable} matching.
\end{thm}

\begin{proof}
See Appendix \ref{app:proof_pairwise_stable}.
\end{proof}

\par A pairwise stable matching implies that both vehicles and time-frequency resources obtain the highest utility under the current conditions. It is important to ensure that the pairwise stable matching results can be obtained in the finite iteration times, which can be proved as follows.

\begin{thm}
\label{thm:convergence}
The DV-RMA converges to a pairwise stable matching after a limited number of iterations.
\end{thm}
\begin{proof}
See Appendix \ref{app:proof_convergence}.
\end{proof}

\subsubsection{Complexity}
We present the complexity of the DV-RMA and the greedy algorithm, described below in this part, in this part. We first analyze the complexity of each round of the DV-RMA, however, the strict upper bound for the number of dynamic rounds of DV-RMA is hard to obtain because of the following reasons. The number of dynamic rounds in DV-RMA is determined by the relationship of the matching structure and adjusted preference list, and the two factors of each vehicle is influenced by that of others. Hence, the strict complexity of DV-RMA is impeded by the peer effects. Here we give the upper bound of the proposal number in DV-RMA.

\begin{thm}
\label{thm:complexity}
The complexity of each round of the DV-RMA is $O((N+M)(K+K_u)^2T^2)$, and the upper bound of the proposal number in DV-RMA is $O(T(K+K_u)(N+M)^{S+1})$.
\end{thm}
\begin{proof}
See Appendix \ref{app:proof_complexity}.
\end{proof}

We then compare the DV-RMA with the greedy algorithm. In the \emph{greedy} algorithm, all vehicles access the time-frequency resource in turn until they are completely matched. Each vehicle occupies the unsaturated resource over which its SINR is the highest every time, in order to increase the possibility of being successfully detected.
Hence, there are $N+M$ vehicles and each one is matched with $S$ time-frequency resources in the greedy algorithm. Each time the {vehicle selects} one resource to access, it needs to search through all resources that haven't been matched with it and compare its SINR over these resources. Hence, the computing times for each vehicle are
\begin{equation}
[(K+K_u)T]+[(K+K_u)T-1]+\dots+[(K+K_u)T-(S-1)]=S\cdot [(K+K_u)T-\frac{S-1}{4}].
\end{equation}
The complexity of the greedy algorithm is
\begin{equation}
O\left((N+M)\cdot S\cdot [(K+K_u)T-\frac{S-1}{4}]\right).
\end{equation}

\subsection{Selection of the Sensitivity Factor $\lambda$} \label{sec:lambda}
In a subframe, denote the maximum and minimum interference radius of the cellular V2X users to VANET users by $L_{MAX}$ and $L_{min}$ respectively. $\lambda$ reflects the sensitivity of the interference of cellular V2X users to VANET users, and also controls the access of the cellular V2X users into the unlicensed spectrum. The value of $\lambda$ is related to the data transmission demand of VANET users. We then analyze how it tunes the performance.
\subsubsection{$\lambda > \frac{1}{\pi L_{min}^2}$}
In this case, neither the V2I user nor the V2V user can access the unlicensed time-frequency resources, since the penalty item in the objective function is large such that the benefit of the added active number of cellular V2X users cannot offset the disadvantage of the increased interference to VANET users.
Hence, cellular V2X users can only utilize the dedicated cellular spectrum.


\subsubsection{$\frac{1}{\pi L_{min}^2} < \lambda \leq \frac{1}{\pi L_{MAX}^2}$}
In this case, part of the V2I users and the V2V users can access the unlicensed time-frequency resources. The V2V user tends to share the same unlicensed time-frequency resource with other neighbouring V2V users, in order to minimize the total interference range to VANET users.
Hence, the unlicensed time-frequency resources are likely to be allocated to cluster-like cellular V2X users. Moreover, vehicles are encouraged to access more dedicated cellular time-frequency resources and fewer unlicensed time-frequency resources as the penalty factor increases.
\subsubsection{$\lambda\leq \frac{1}{\pi L_{MAX}^2} $}
Cellular V2X users are able to utilize both dedicated cellular and unlicensed time-frequency resources freely in this case. The interference to VANET users does not affect the allocation of the unlicensed time-frequency resources, since the interference to VANET users weighs little in the objective function.

\section{Simulation Results \label{sec:simulation_results}}%

In this section, we evaluate the performance of the proposed DV-RMA in the urban scenario. We compare the number of total active cellular V2X users in two modes: 1) \emph{shared mode} in which both dedicated cellular and unlicensed spectrum are used; 2) \emph{dedicated mode} in which only dedicated cellular spectrum is used. In addition, the proposed DV-RMA is compared with the greedy algorithm, as mentioned in Section \ref{sec:complexity}-2, in terms of the system performance.

\begin{table}
\small
\renewcommand{\arraystretch}{1.3}
\caption{Parameters Setting for Simulation.}
\label{table:parameters}
\centering
\begin{tabular}{l|l}
\hline
Parameters & Values \\
\hline
Transmit power of vehicles $P^v$ & $23dbm$\\
Threshold of the receiving power $P^r$ & $-75dbm$\\
Noise power spectrum density & $-174dbm/Hz$\\
Power gain factor $G$ & $-31.5db$ \\
Accessible threshold $\gamma_{th}$ & $0dB$ \\
Fading factor $\alpha$ & 3\\
Carrier center frequency & $2.4GHz$ \\
Subchannel bandwidth $B$ & $10kHz$ \\
Number of dedicated cellular subchannels $K$ & $10$ \\
Number of unlicensed subchannels $K_u$ & $10$ \\
Maximum number of time-frequency resources $Q$ & $3$ \\
Maximum number of vehicles $S$ & $3$ \\
Number of subframes in each SPS cycle $T$ & $10$ \\
\hline
\end{tabular}
\end{table}

\begin{figure}
\centering
\includegraphics[width=0.65\textwidth]{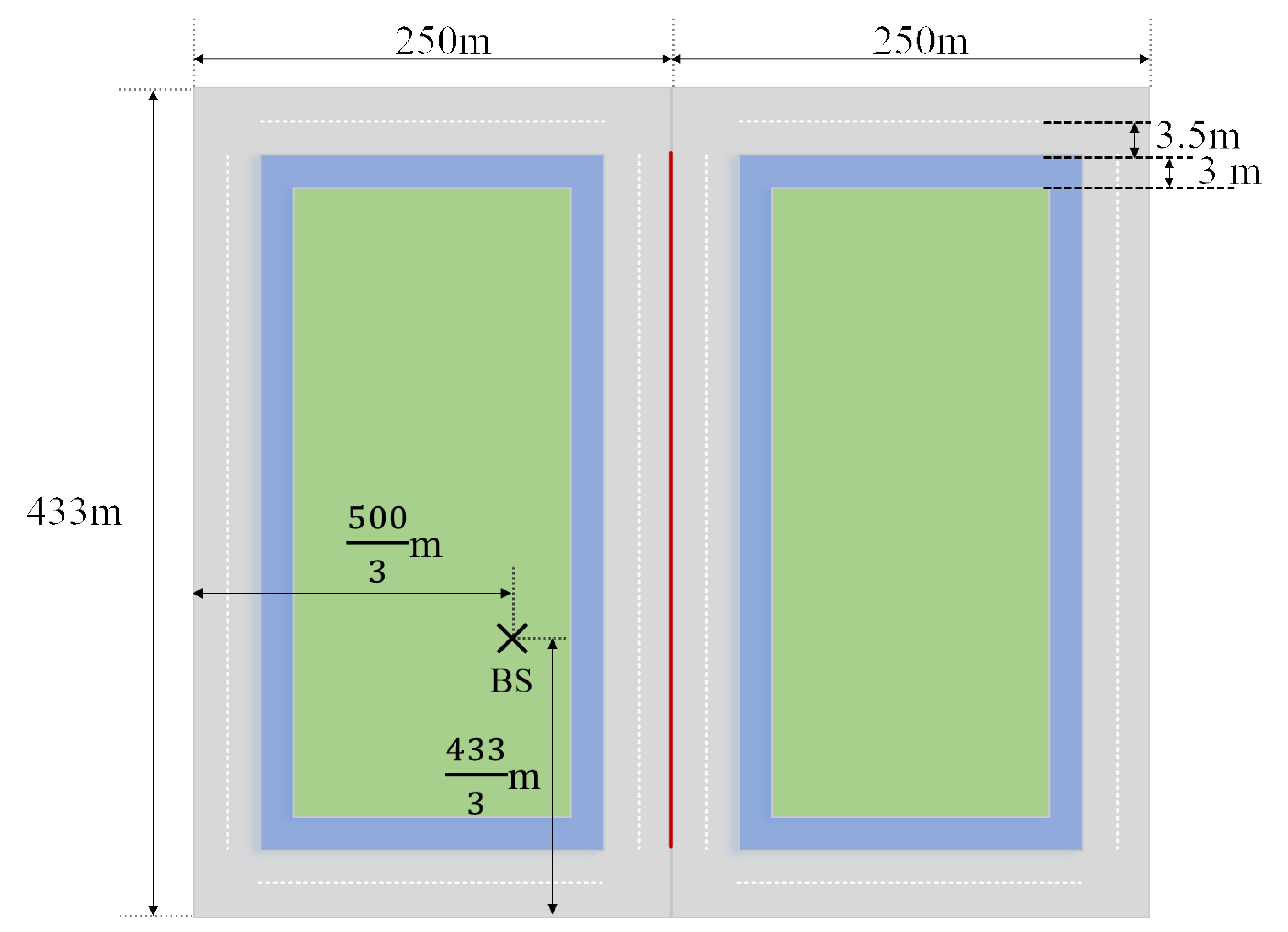}
\caption{Road setting in the urban scenario.}
\label{Fig:urban}
\end{figure}

We discuss the urban scenario defined in [\ref{ref:3GPP_V2X}], which is presented in the Fig. \ref{Fig:urban}. Both cellular V2X users and VANET users are randomly located on the road. The safe inter-vehicle distance in the same lane is $2.5s\times v$, where the speed of vehicles $v$ is $15 \sim 60km/h$ according to [\ref{ref:3GPP_V2X}]. Other parameters of the simulation are shown in the Table \ref{table:parameters} \cite{BJBF-2007}.

\begin{figure}
\centering
\includegraphics[width=0.65\textwidth]{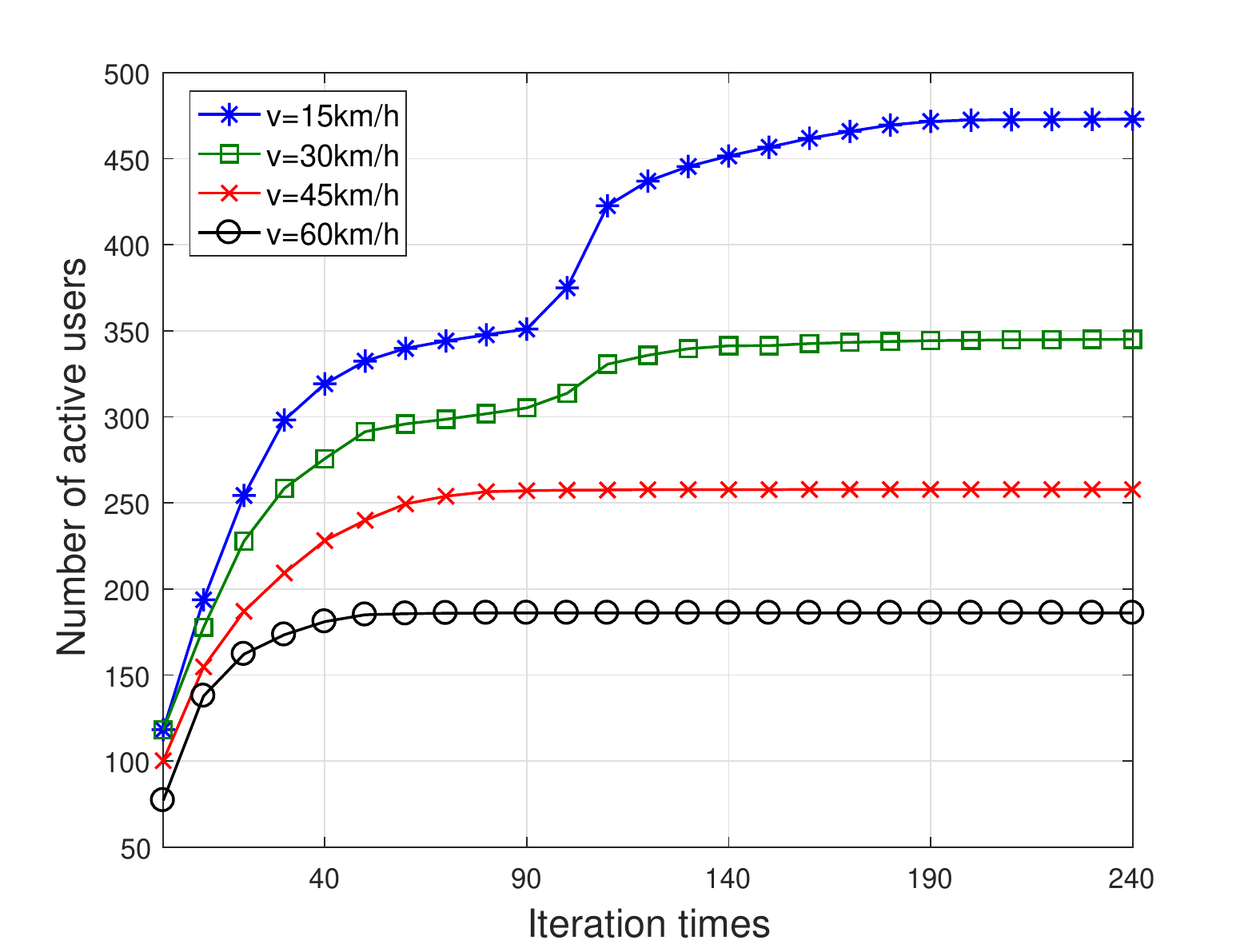}
\caption{Number of active cellular V2X users v.s. the iteration times with different velocities.}
\label{Fig:unlicensed_iteration}
\end{figure}

Fig. \ref{Fig:unlicensed_iteration} shows the number of active cellular V2X users in the SPS cycle vs. the iteration times with the speed of vehicles from $15km/h$ to $60km/h$. 
The number of active cellular V2X users decreases as the speed increases, because the number of vehicles decreases for the larger safe inter-vehicle {distance.}
We also observe that the number of active cellular V2X users converges to a stable value as that of the iterations increases, e.g., in DV-RMA the number of active cellular V2X users converges after 200 cycles, thereby reflecting the convergence of the DV-RMA.

\begin{figure}
\centering
\includegraphics[width=0.65\textwidth]{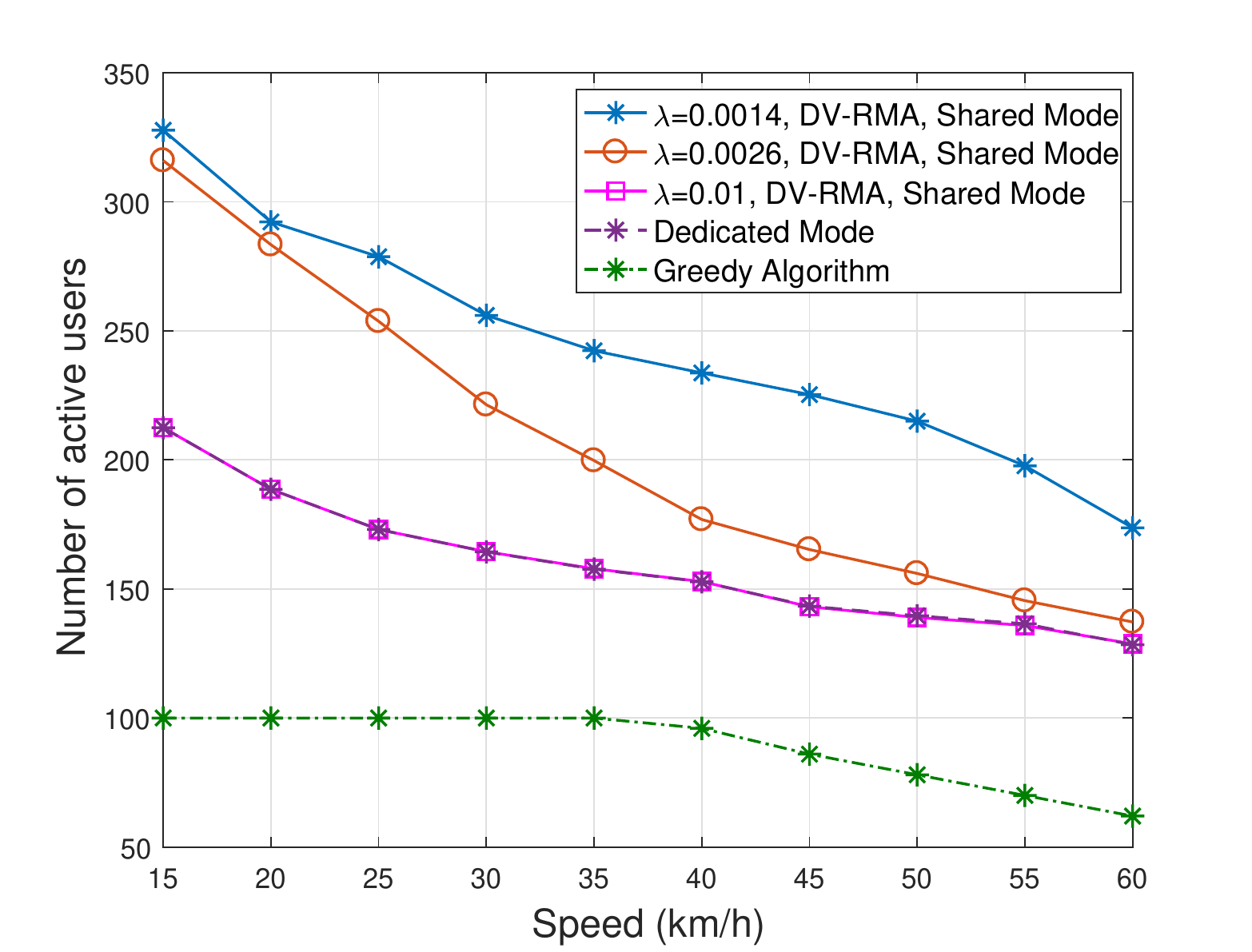}
\caption{Number of active cellular V2X users v.s. speed in different modes.}
\label{Fig:licensed}
\end{figure}

Fig. \ref{Fig:licensed} presents the number of active cellular V2X users v.s. the speed in the shared mode and dedicated mode. We can learn that all the curves decrease with the increasing speed. In the shared mode, the number of active cellular V2X users is no fewer than that in the dedicated mode. The gap between curves in the two modes increases as the speed decreases, because the limited dedicated cellular time-frequency resources cannot satisfy the transmission demand of vehicles, while the unlicensed resources can serve as the complement to support more vehicles.
Besides, as the penalty factor $\lambda$ increases, the number of active cellular V2X users decreases. Since the interference to VANET users caused by cellular V2X users becomes severer as the penalty factor $\lambda$ increases, the BS is discouraged to allocate the unlicensed spectrum to the cellular V2X users.
It is also observed that the number of active cellular V2X users in the DV-RMA is at least twice larger than that in the greedy algorithm, e.g. when $\lambda=0.0026$ and the speed is $15km/h$ in the DV-RMA, the number of active cellular V2X users is more than 300, while that is only 100 in the greedy algorithm. The DV-RMA performs better than the greedy algorithm because vehicles dynamically change their preference lists and resources can be allocated to the vehicles which bring them better performance.

\begin{figure}
\centering
\includegraphics[width=0.65\textwidth]{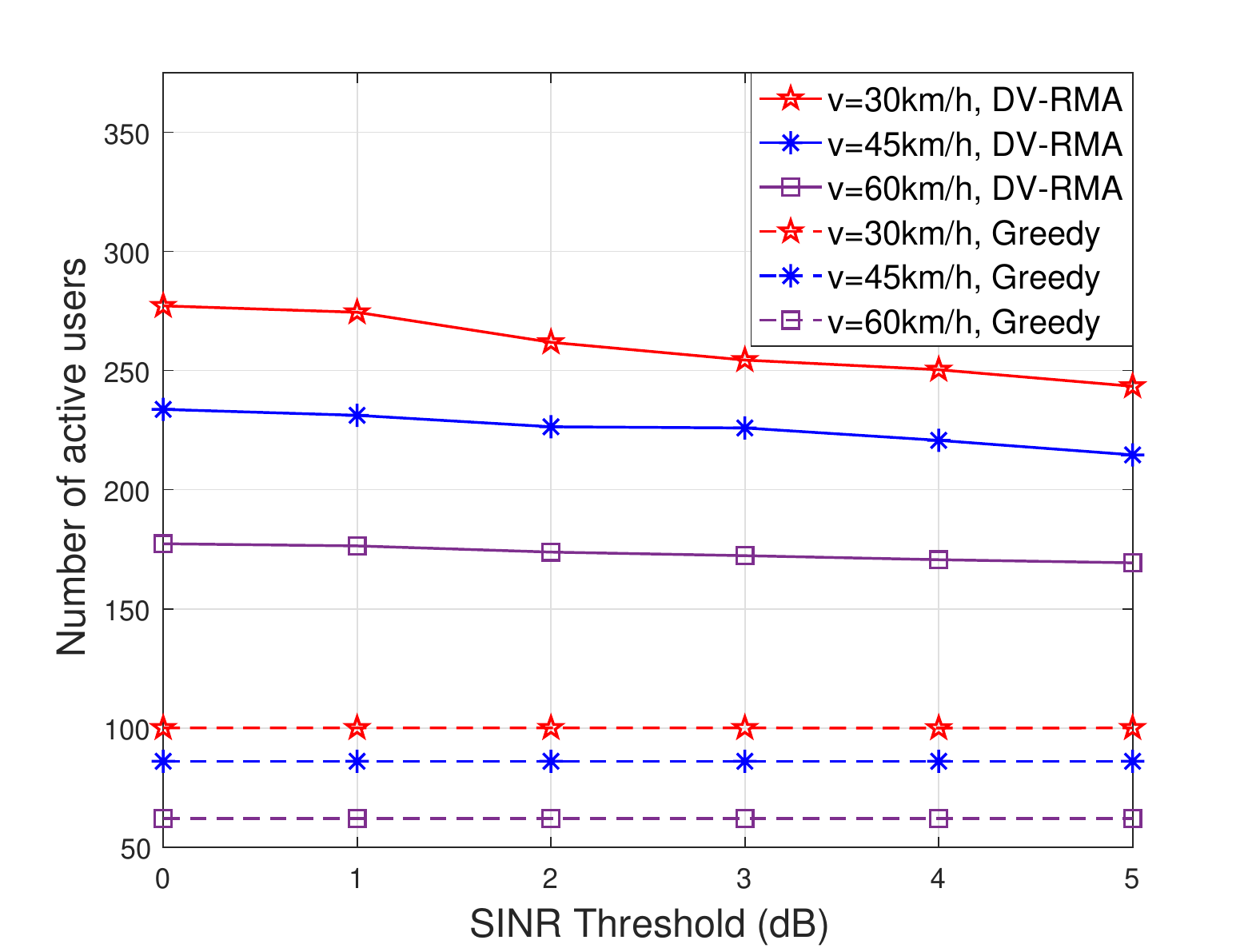}
\caption{Number of active cellular V2X users vs. active SINR threshold $\gamma_{th}$ with different speed.}
\label{Fig:threshold}
\end{figure}

Fig. \ref{Fig:threshold} illustrates the number of active cellular V2X users v.s. the SINR threshold $\gamma_{th}$ with different speed.
It is easily observed that the DV-RMA obtains better performance than the greedy algorithm, even when the speed in the DV-RMA is higher than that in the greedy algorithm.
In the DV-RMA, the number of active cellular V2X users decreases as the SINR threshold increases. This implies that the higher the threshold is, the harder it is for cellular V2X users to successfully detect over the unlicensed resources.
In the greedy algorithm, the allocation of time-frequency resources is only related to the rank of the SINR of the cellular V2X user. Hence, the number of active cellular V2X users remains the same as the SINR threshold increases.
Moreover, in the DV-RMA, the number of active cellular V2X users grows as the speed decreases with the same SINR threshold, since the number of vehicles increases with the smaller inter-vehicle distance. 

\begin{figure}
\centering
\includegraphics[width=0.65\textwidth]{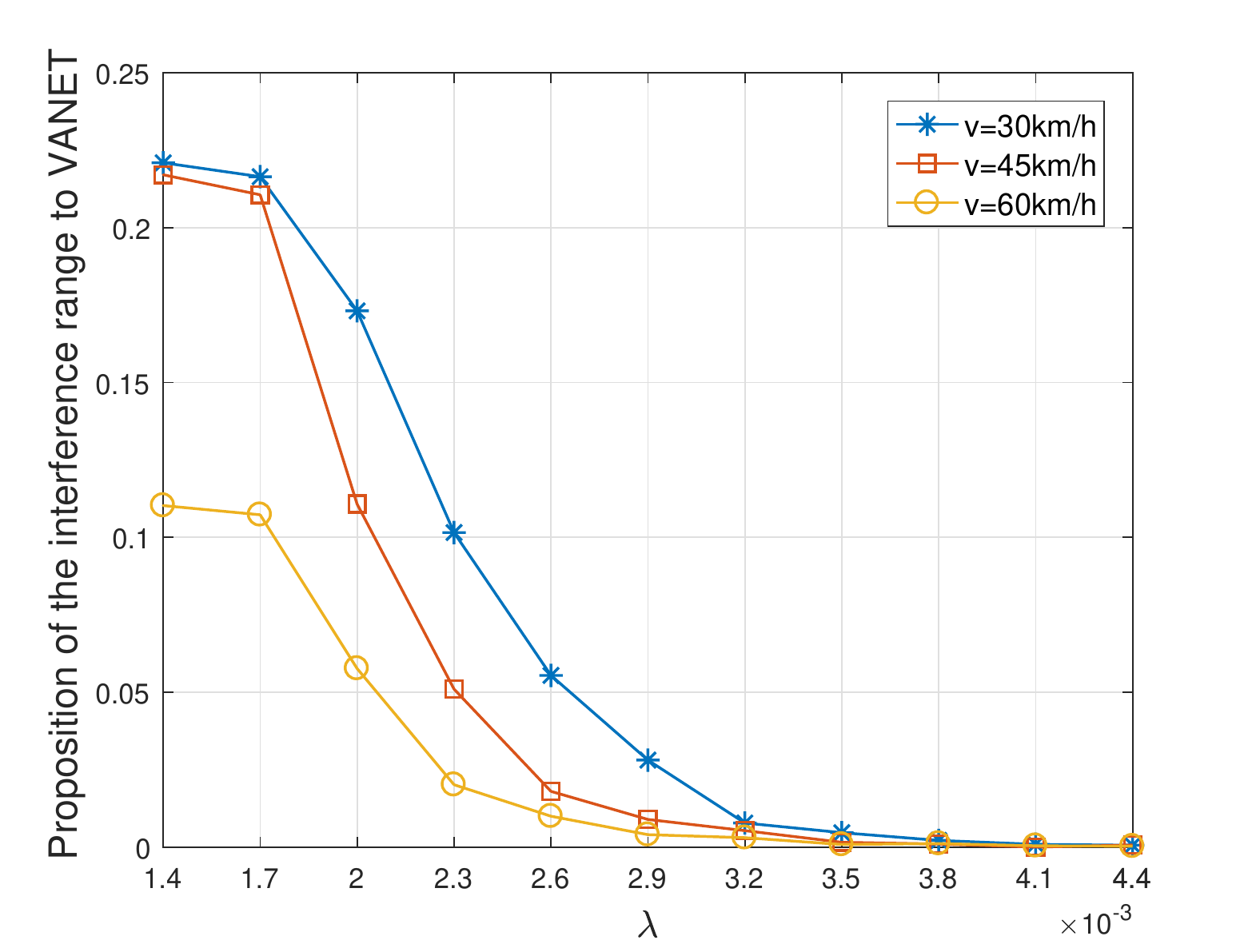}
\caption{Proportion of total interference range to VANET and the sum area vs. penalty factor $\lambda$ with different speed.}
\label{Fig:interference}
\end{figure}

Fig. \ref{Fig:interference} presents the interference range to VANET users v.s. the penalty factor $\lambda$ with different speed.
We evaluate the interference range by the ratio of the total interference range to the sum area.
We observe that as the penalty factor $\lambda$ increases, the interference range to VANET users decreases. This is because unlicensed resources are discouraged to be allocated to cellular V2X users with large penalty factor. Moreover, the interference range is larger with lower speed, which implies that more transmission demands need to be offloaded to unlicensed spectrum and the interference to VANET users is severer with lower speed.

\begin{figure}
\centering
\includegraphics[width=0.65\textwidth]{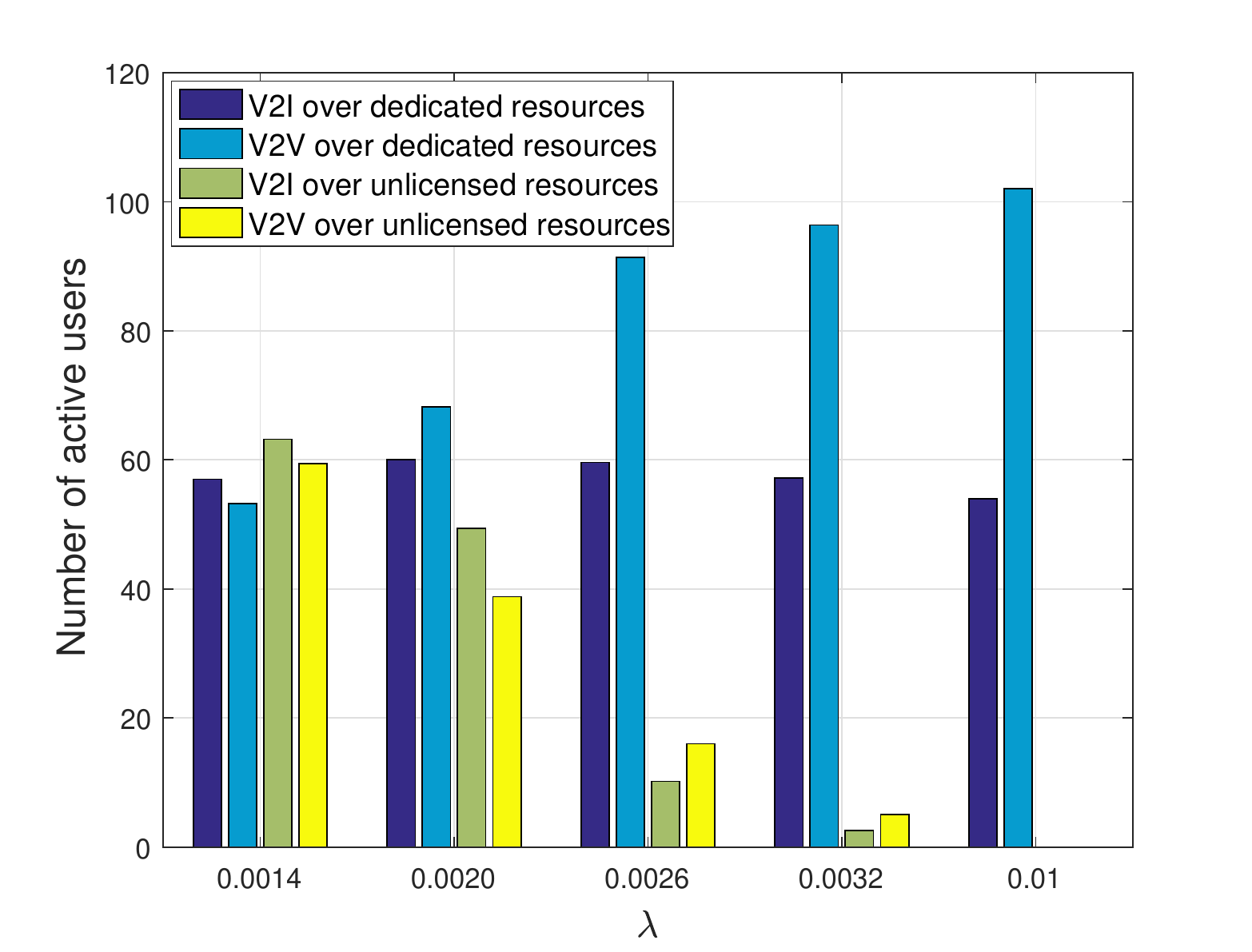}
\caption{Number of active V2I and V2V over dedicated and unlicensed resources vs. penalty factor $\lambda$ when $v=45km/h$.}
\label{Fig:res_lambda}
\end{figure}

Fig. \ref{Fig:res_lambda} shows the number of active V2I and V2V users over dedicated and unlicensed resources v.s. the penalty factor $\lambda$ when the speed is $45km/h$. It can be observed that the unlicensed subchannel for cellular V2X users decrease as the value of $\lambda$ increases, in order to reduce the interference to VANET users. Cellular V2X users can utilize both dedicated cellular and unlicensed resources freely when $\lambda < \frac{1}{\pi L_{MAX}^2}$, 
and are discouraged to access the unlicensed resources as $\lambda$ increases. When $\lambda$ is larger than a threshold, e.g., $\lambda > 0.01$, cellular V2X users are not allowed to utilize the unlicensed resources. This is consistent with the results analyzed in Section \ref{sec:lambda}.


\section{Conclusion \label{sec:conclusion}}%

In this paper, we studied the spectrum sharing problem where cellular V2X users coexist with VANET users in the unlicensed spectrum. We proposed an energy sensing based spectrum sharing scheme for cellular V2X users to share the unlicensed spectrum fairly with VANET users. The allocation problem of time-frequency resources to cellular V2X users in the semi-persistent scheduling cycle was formulated as a two-sided many-to-many matching problem with peer effects. We then developed a dynamic vehicle-resource matching algorithm to solve the problem, and analyzed its stability, computational complexity, as well as convergence.
Simulation results showed that enabling the sharing of the unlicensed spectrum can increase the system performance, and the proposed DV-RMA obtained better performance than the greedy algorithm.
The communication demands of different types of users can be satisfied by adjusting the penalty factor. More cellular V2X users can be  supported over the unlicensed spectrum as the penalty factor decreases.


\vspace{-3mm}
\begin{appendices}
\vspace{-2mm}

\section{Proof of Proposition 1}\label{app:proof_prop}

The overlapping area of cellular V2X users $i$ and $j$ is related to the distance $d_{i,j}$ between them. Here we assume that $L_i<L_j$, and the analysis is similar when $L_i<L_j$.
{
\par When $d_{i,j}>\sqrt{L_j^2-L_i^2}$, as shown in Fig.~\ref{Fig:overlap}-(a), the overlapping area can be expressed as follows, as described in \cite{circle_intersection}.
\begin{equation}
\label{e:cal_overlap_a}
O_{i,j}=L_i^2 \arccos(\frac{x_i}{L_i})-x_i \sqrt{L_i^2-x_i^2} + L_j^2 \arccos(\frac{x_j}{L_j})-x_j \sqrt{L_j^2-x_j^2},
\end{equation}
where
\begin{equation}
\label{e:cal_overlap_a_x}
x_i=\frac{ d_{i,j}^2+L_i^2-L_j^2}{2d_{i,j}}, \ x_j=\frac{ d_{i,j}^2-L_i^2+L_j^2}{2d_{i,j}}.
\end{equation}
}
\par When $d_{i,j}\leq\sqrt{L_j^2-L_i^2}$, as shown in Fig.~\ref{Fig:overlap}-(b), {which is similar to the case in \cite{circle_intersection},} we have
\begin{equation*}
\left\{
\begin{array}{ll}
x_i+d_{i,j} = x_j ;\\
L_i^2-x_i^2 = L_j^2-x_j^2.
\end{array}
\right.
\end{equation*}
Hence,
\begin{equation}
\label{e:cal_overlap_b_x}
x_i=\frac{ L_j^2 - d_{i,j}^2-L_i^2}{2d_{i,j}}, \ x_j=\frac{ d_{i,j}^2-L_i^2+L_j^2}{2d_{i,j}}.
\end{equation}
The overlapping area can be expressed as
\begin{equation}
\label{e:cal_overlap_b}
O_{i,j}=\pi L_i^2 - \left(L_i^2 \arccos(\frac{x_i}{L_i})-x_i \sqrt{L_i^2-x_i^2}\right) + L_j^2 \arccos(\frac{x_j}{L_j})-x_j \sqrt{L_j^2-x_j^2}.
\end{equation}

We replace $x_i$ and $x_j$ in (\ref{e:cal_overlap_a}) and (\ref{e:cal_overlap_b}) with (\ref{e:cal_overlap_a_x}) and (\ref{e:cal_overlap_b_x}), then we have (\ref{e:overlap}). $\hfill\blacksquare$

\begin{figure}
\centering
\includegraphics[width=0.75\textwidth]{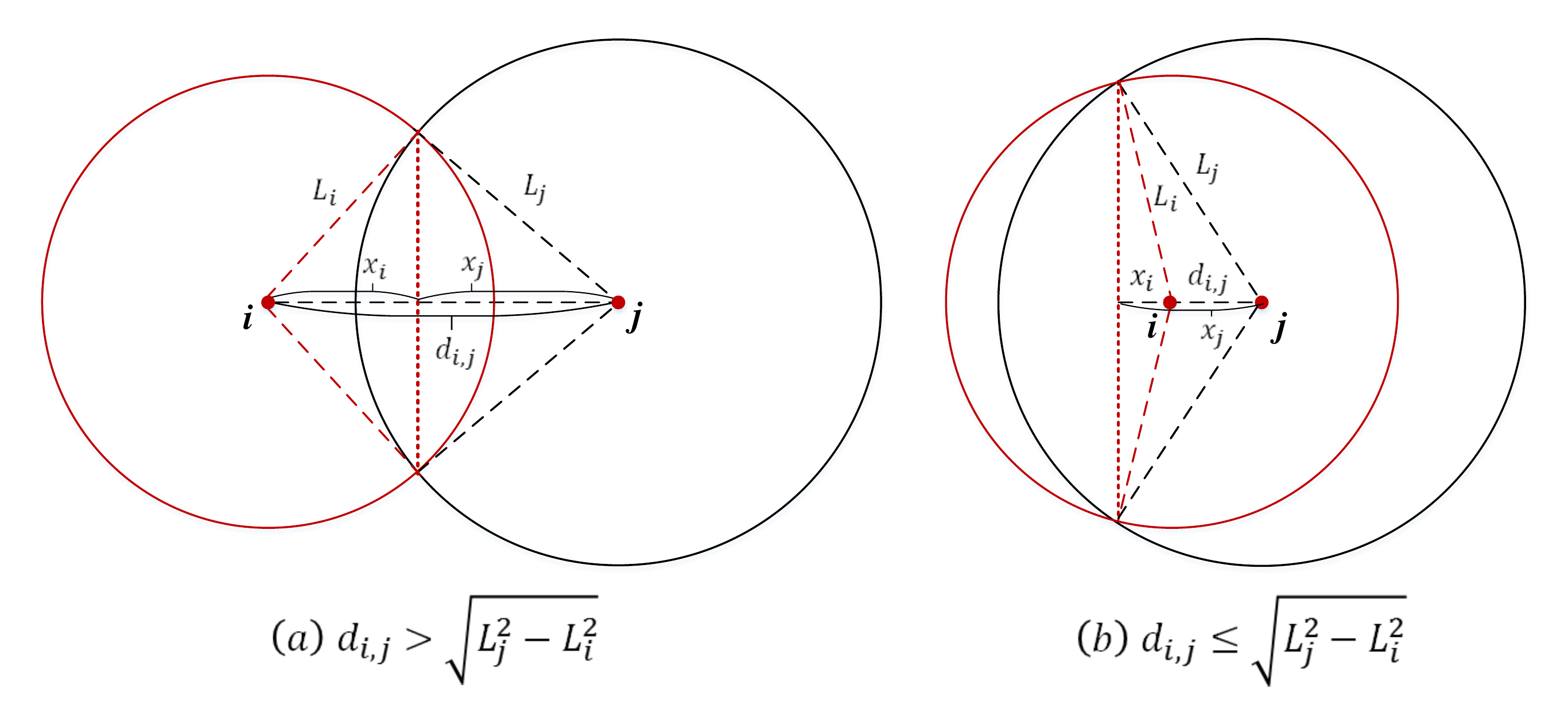}
\caption{Two cases of calculating the overlapping area of the interference range.}
\label{Fig:overlap}
\end{figure}

\section{Proof of Theorem \ref{thm:pairwise stable}}\label{app:proof_pairwise_stable}
 If $\Psi$ is not a \emph{pairwise stable} matching, there exists a pair $(V_i,W_{k,t})$ such that $A\succ_{W_{k,t}}\Psi(W_{k,t}), A \subseteq \{V_i\}\cup\Psi(W_{k,t}), V_i\in A,$ and $W_{k,t}\succ_{V_i}W_{k',t'}, W_{k',t'}\in\Psi(V_i)$.
Apparently, $V_i$ will propose itself to $W_{k,t}$ in the next round of \emph{matching process} to pursue its higher interest. However, since the algorithm has converged to $\Psi$, there are no more proposals of vehicles. The only reason why $V_i$ does not propose itself to $W_{k,t}$ is that $\Psi(W_{k,t})$ is the incompatible list to $V_i$. This means that $W_{k,t}$ rejected $V_i$ before under the same matching structure, denoted by $\Psi(W_{k,t})\succ_{W_{k,t}}A', \ A' \subseteq \{V_i\}\cup\Psi(W_{k,t}),$ which is contradictory to the pairwise stable assumption. Hence, \textbf{theorem \ref{thm:pairwise stable}} is proved. $\hfill\blacksquare$

\section{Proof of Theorem \ref{thm:convergence}}\label{app:proof_convergence}
With the definition of incompatible list in the DV-RMA algorithm, there are only two conditions that vehicle $V_i$ proposes to time-frequency resource $W_{k,t}$. One is that $W_{k,t}$ accept the proposal without rejecting any other vehicle. The other is that $W_{k,t}$ rejects one of $V_i$ and formerly accepted vehicles, and then adds the updated $\Psi(W_{k,t})$ into the incompatible list of the rejected vehicle. Since one vehicle is matched with at most $S$ time-frequency resources, the times of accepting without rejecting any other vehicle is limited in $(N+M)S$. Moreover, each vehicle has no more than $(C_{N+M-1}^{1}+C_{N+M-1}^{2}+\dots+C_{N+M-1}^{S})$ incompatible lists over each time-frequency resources, so the total incompatible lists is no more than
\begin{equation}
\label{e:incompatible_pair}
T(K+K_u)(C_{N+M-1}^{1}+C_{N+M-1}^{2}+\dots+C_{N+M-1}^{S}).
\end{equation}
Hence, the total number of matched pairs and incompatible lists is limited and the DV-RMA converges to a \emph{pairwise stable} matching after a limited number of iterations. $\hfill\blacksquare$

\section{Proof of Theorem \ref{thm:complexity}}\label{app:proof_complexity}
Each dynamic round of DV-RMA includes preference list updating and \emph{matching process}.
Updating a preference list is actually a sorting process of preferences with the complexity of $O\left((K+K_u)^2T^2\right)$. At most $N+M$ vehicles need to update their preference lists. Hence, the complexity of preference list updating is $O\left((N+M)(K+K_u)^2T^2\right)$.
In the \emph{matching process}, at most $N+M$ vehicles propose to time-frequency resources in each round, and one vehicle proposes itself to no more than $(K+K_u)T$ time-frequency resources during the \emph{matching process}. Hence, the complexity of \emph{matching process} is $O\left((N+M)(K+K_u)T\right)$.
Therefore, the complexity of each dynamic round of the DV-RMA is
\begin{equation}
O((N+M)(K+K_u)^2T^2)+O((N+M)(K+K_u)T)= O((N+M)(K+K_u)^2T^2).
\end{equation}

\par The DV-RMA includes the \emph{initialization phase} and \emph{matching phase}. The total number of matching pairs initialized in the \emph{initialization phase} is no more than $(N+M)S$. In the \emph{matching phase}, the maximum proposal times is described as (\ref{e:incompatible_pair}) in \textbf{theorem \ref{thm:convergence}}. Hence, the upper bound of the proposal number in DV-RMA is
\begin{small}
\begin{equation*}
O\left((N+M)S\right)+O\left(T(K+K_u)(C_{N+M-1}^{1}+C_{N+M-1}^{2}+\dots+C_{N+M-1}^{S})\right)=O\left(T(K+K_u)(N+M)^{S+1}\right).
\end{equation*}
\end{small} $\hfill\blacksquare$

\end{appendices}


\vspace{-5mm}


\begin{thebibliography}{20}
\bibitem{KAEH-2011}
\label{ref:ITS_V2X}
G. Karagiannis, O. Altintas, E. Ekici, G. Heijenk, B. Jarupan, K. Lin, and T. Weil, ``Vehicular networking: A survey and tutorial on requirements, architectures, challenges, standards and solutions'', \emph{IEEE Commun. Surveys $\&$ Tutorials}, vol. 13, no. 4, pp. 584-615, Sept. 2011.

\bibitem{SNHH-2015}
\label{ref:D2D}
L. Song, D. Niyato, Z. Han, and E. Hossain, \emph{Wireless device-to-device communications and networks}, Cambridge University Press, UK, 2015.

\bibitem{SS-2010}
\label{ref:LTE}
L. Song and J. Shen, \emph{Evolved network planning and optimization for UMTS and LTE}, Auerbach Publications, CRC Press, 2010.

\bibitem{FA-2015}
F. Andreas, ``Standards for vehicular communication-from IEEE 802.11p to 5G'', \emph{Elektrotechnik and Informationstechnik} vol. 132, no. 7, pp. 409-416, Sept. 2015.

\bibitem{ISTM-2016}
I. Safiulin, S. Schwarz, T. Philosof and M. Rupp, ``Latency and resource utilization analysis for V2X communication over LTE MBSFN transmission'', in \emph{Proc. WSA}, Munich, Mar. 2016.


\bibitem{3GPP-REL14}
3GPP TR 22.885, ``Study on LTE support for V2X services'', Release 14, Sept. 2015.

\bibitem{BLY-2017}
B. Di, L. Song, Y. Li, and Z. Han, ``V2X meets NOMA: Non-orthogonal multiple access for 5G enabled vehicular networks", \emph{IEEE Wireless Commun.}, under revision.

\bibitem{K-2016}
K. Higuchi, ``NOMA for future cellular systems'', in \emph{Proc. IEEE VTC}, Montreal, Sept. 2016.

\bibitem{BSCS-2013}
M. Bennis, M. Simsek, A. Czylwik, W. Saad, S. Valentin, and M. Debbah, ``When cellular meets WiFi in wireless small cell networks?'', \emph{IEEE Commun. Mag.}, vol. 51, no. 6, pp. 44-50, Jun. 2013.

\bibitem{WSMK-2011}
W. Kim, S. Y. Oh, M. Gerla and K. C. Lee, ``CoRoute: A new cognitive anypath vehicular routing protocol'', in \emph{ Proc. IWCMC}, Istanbul, Jul. 2011.

\bibitem{JYMJ-2016}
J. M.-Y. Lim, Y. C. Chang, M. Y. Alias and J. Lo, ``Cognitive radio network in vehicular ad hoc network (VANET): A survey'', \emph{Cogent Engineering}, vol. 3, no. 1,  pp. 1-19, Jun. 2016.

\bibitem{SPS-2015}
3GPP R1-157449, ``Further discussion on resource allocation mechanism in PC5-based V2V'', Nov. 2015.

\bibitem{M-2013}
\label{ref:matching_preference}
D. Manlove, ``Algorithmics of matching under preferences'', \emph{World Scientific}, 2013.

\bibitem{DBSL-2015}
\label{ref:matching_dby}
B. Di, L. Song, and Y. Li, ``Sub-channel assignment, power allocation, and user scheduling for non-orthogonal multiple access networks", \emph{IEEE Trans. Wireless Commun.}, vol. 16, no. 11, pp. 7686-7698, Nov. 2016.

\bibitem{ZDSL-2016}
\label{ref:matching_zsh}
S. Zhang, B. Di, L. Song, and Y. Li. ``Sub-channel and power allocation for non-orthogonal multiple access relay networks with amplify-and-forward protocol'', \emph{IEEE Trans. Wireless Commun.}, vol. 16, no. 4, pp. 2249-2261, Apr. 2017.

\bibitem{ZLS-2016}
\label{ref:D2D-U}
H. Zhang, Y. Liao and L. Song, ``D2D-U: Device-to-device communications in unlicensed bands'', \emph{IEEE Trans. Wireless Commun.}, to appear.

\bibitem{BLCHW-2011}
\label{ref:matching_peer_effect}
E. Baron, C. Lee, A. Chong, B. Hassibi, and A. Wierman, ``Peer effects and stability in matching markets'', in \emph{ Proc. SAGT}, Amalfi, Oct. 2011.

\bibitem{LJMC-2013}
L.-C. Tung, J. Mena, M. Gerla, and C. Sommer, ``A cluster based architecture for intersection collision avoidance using heterogeneous networks'', in \emph{Proc. MED-HOC-NET}, Ajaccio, Jun. 2013.

\bibitem{SLCH-2016}
S. Gopinath, L. Wischhof, C. Ponikwar and H. J. Hof, ``Hybrid solutions for data dissemination in vehicular networks'', in \emph{Proc. WD}, Toulouse, Mar. 2016.

\bibitem{CYSM-2016}
\label{ref:data_offloading3}
Q. Chen, G. Yu, H. Shan, A. Maaref, G. Y. Li and A. Huang, ``Cellular meets WiFi: Traffic offloading or resource sharing?'', \emph{IEEE Trans. Wireless Commun.}, vol. 15, no. 5, pp. 3354-3367, May 2016.
\bibitem{3GPP-2013}
\label{ref:3GPP_LTE-U}
3GPP TR 36.808, ``Evolved universal terrestrial radio access (E-UTRA); carrier aggregation; base station (BS) radio transmission and reception'', Release 10, Jul. 2013.





{
\bibitem{Rayleigh}
T. S. Rappaport, \emph{Wireless communications: principles and practice.} Englewood Cliffs, NJ: Prentice-Hall, 1996.
}
\bibitem{QUAL-2014}
\label{ref:LTE-U_Qualcomm}
Qualcomm research, ``LTE in unlicensed spectrum: Harmonious coexistence with Wi-Fi'', Jun. 2014.

\bibitem{BAH-2014}
\label{ref:LTE-U_dutycycle}
A. Babaei, J. Andreoli-Fang, and B. Hamzeh, ``On the impact of LTE-U on Wi-Fi performance'', in \emph{Proc. PIMRC}, Washington, Sept. 2014.

\bibitem{RS-1992}
\label{ref:GS_algorithem}
A. Roth and M. Sotomayor, \emph{Two-sided matching: A study in game theoretic modeling and analysis}, Cambridge University Press, UK, 1992.

\bibitem{3GPP-2016}
\label{ref:3GPP_V2X}
3GPP RP-161894, ``LTE-based V2X services'', Sept. 2016.

\bibitem{BJBF-2007}
B. J. B. Fonseca, ``A distributed procedure for carrier sensing threshold adaptation in CSMA-based mobile ad hoc networks", in \emph{Proc. IEEE VTC}, Baltimore, Sept. 2007.


{
\bibitem{circle_intersection}
Weisstein, Eric W. "Circle-Circle Intersection." From MathWorld--A Wolfram Web Resource. http://mathworld.wolfram.com/Circle-CircleIntersection.html
}
\end{thebibliography}
\end{document}